\documentclass[11pt, reqno]{amsart}
\usepackage{amsfonts}
\usepackage{amsmath}
\usepackage{amssymb}
\usepackage{amsthm}
\usepackage{subfig}
\usepackage[mathscr]{eucal}
\usepackage[pdftex]{graphicx, color}
\usepackage{amscd}
\usepackage[numbers, sort&compress]{natbib}
\usepackage{url}
\usepackage[pdftex, colorlinks=true]{hyperref}
\hypersetup{pdfauthor={Jason Dominy and Herschel Rabitz}, pdftitle={Volume Fractions of the Kinematic ``Near-Critical'' Sets of the Quantum Ensemble Control Landscape}}

\author[Dominy and Rabitz]{Jason Dominy$^{1}$ and Herschel Rabitz$^{2,1}$\\$^{1}$ Program in Applied and Computational Mathematics, Princeton University\\$^{2}$ Department of Chemistry, Princeton University}
\title[Volume Fractions of ``Near-Critical'' Sets]{Volume Fractions of the Kinematic ``Near-Critical'' Sets of the Quantum Ensemble Control Landscape}
\DeclareMathOperator{\Vol}{Vol}
\DeclareMathOperator{\VolFrac}{VolFrac}
\DeclareMathOperator{\Area}{Area}
\DeclareMathOperator{\Tr}{Tr}

\DeclareMathOperator{\Image}{Im}

\DeclareMathOperator{\grad}{grad}
\DeclareMathOperator{\Hess}{Hess}
\DeclareMathOperator{\Crit}{Crit}

\DeclareMathOperator{\diag}{diag}
\DeclareMathOperator{\Orb}{Orb}
\DeclareMathOperator{\Stab}{Stab}

\DeclareMathOperator{\ad}{ad}

\DeclareMathOperator{\Ellipse}{Ellipse}

\newtheorem{theorem}{Theorem}
\newtheorem{lemma}{Lemma}

\newtheorem{example}{Example}
\newtheorem{conjecture}{Conjecture}
\begin{document}
\date{ \today}
\thanks{\emph{E-mail address}: {\tt jdominy@princeton.edu}, {\tt hrabitz@princeton.edu}}
%\numberwithin{equation}{section}
\newcommand{\jmdstack}[2]{\genfrac{}{}{0pt}{}{#1}{#2}}
\newcommand{\rmd}{\mathrm{d}}
\newcommand{\rmT}{\mathrm{T}}
\newcommand{\rmHS}{\mathrm{HS}}
\newcommand{\UN}{\mathrm{U}(N)}
\newcommand{\uN}{\mathrm{u}(N)}
\newcommand{\Un}{\mathrm{U}(\mathbf{n})}
\newcommand{\un}{\mathrm{u}(\mathbf{n})}
\newcommand{\Um}{\mathrm{U}(\mathbf{m})}
\newcommand{\um}{\mathrm{u}(\mathbf{m})}
\newcommand{\UK}{\mathrm{U}(\mathbf{K})}
\newcommand{\uK}{\mathrm{u}(\mathbf{K})}
\newcommand{\Umn}{\Um\oplus\Un}
\newcommand{\umn}{\um\oplus\un}
\newcommand{\gZ}{\mathfrak{g}_{\mathbb{Z}}}
\newcommand{\tZ}{\mathfrak{t}_{\mathbb{Z}}}
\newcommand{\Ob}{\mathcal{O}}
\newcommand{\perm}{\mathscr{P}}
\newcommand{\cG}{\mathcal{G}}
\newcommand{\cH}{\mathcal{H}}
\newcommand{\cK}{\mathcal{K}_{\perm}}
\newcommand{\cA}{\mathcal{A}}
\newcommand{\cB}{\mathcal{B}}
\newcommand{\cC}{\mathcal{C}}
\newcommand{\CtlSp}{\mathbb{K}}
\newcommand{\smd}{d}

\begin{abstract}An estimate is derived for the volume fraction of a subset $C_{\epsilon}^{\perm} = \{U \;:\; \|\grad J(U)\|\leq \epsilon\}\subset\UN$ in the neighborhood of the critical set $C^{\perm}\simeq\mathrm{U}(\mathbf{n})\perm\mathrm{U}(\mathbf{m})$ of the kinematic quantum ensemble control landscape $J(U) = \Tr(U\rho U^{\dag}\mathcal{O})$, where $U$ represents the unitary time evolution operator, $\rho$ is the initial density matrix of the ensemble, and $\Ob$ is an observable operator.  This estimate is based on the Hilbert-Schmidt geometry for the unitary group and a first-order approximation of $\|\grad J(U)\|^{2}$.  An upper bound on these near-critical volumes is conjectured and supported by numerical simulation, leading to an asymptotic analysis as the dimension $N$ of the quantum system rises in which the volume fractions of these ``near-critical'' sets decrease to zero as $N$ increases.  This result helps explain the apparent lack of influence exerted by the many saddles of $J$ over the gradient flow.
\end{abstract}

\maketitle

\section{Introduction}

%\begin{itemize}
%\item Brief explanation of landscapes, both kinematic and dynamical, building to the statement of ``no traps''
%\item The potential for saddles to play a role impeding the gradient flow of $J$ (or $\tilde{J}$).
%\item Volume fraction of near critical sets as a proxy for probability of being near a saddle -- near enough for the small gradient to be felt.
%\item Numerical studies have not observed any significant ``saddle effect''.  Why?  Asymptotic analysis
%\end{itemize}

Control landscapes are proving to be valuable for providing insights into quantum optimal control theory \cite{Chakrabarti2007}.  A simplification can be achieved by observing that the dynamical landscape $\tilde{J}:\CtlSp\to\mathbb{R}$ (i.e. the objective functional) -- a map taking a control function as input and producing the value of the observable at some final time $T$ -- can be written as a composition of a kinematic landscape $J$ and a control$\mapsto$propagator map $V_{T}$.  Here the kinematic landscape $J:\UN\to\mathbb{R}$ is a smooth real-valued function on the unitary group, and $V_{T}:\mathbb{K}\to\UN$ is defined implicitly through the Schr\"odinger equation and returns the unitary time evolution operator at the final time $T$ for each given input control function.  The goal of quantum optimal control is generally to maximize the dynamical landscape $\tilde{J} = J\circ V_{T}$.

%
%{\bf{The last few years have seen the concept of control landscapes used in connection with a variety of problems in quantum control theory \cite{Chakrabarti2007}.  Generally, a control landscape is a map, defined implicitly through the control dynamical system, that assigns a real observable value to each admissible control field.  In the case of quantum systems, the landscapes may be conveniently thought of as the composition of two maps: a control$\to$state map and a state$\to$observable map.  The ``state'' in this case is frequently the unitary time-evolution operator $U(T,0)$ which is the general solution to the Schr\"odinger equation at some final time $T$; so the control$\to$state map takes in a control field and a final time $T$ and produces $U(T,0)$.  The state$\to$observable map (often called the ``kinematic landscape'') is a real-valued function on the special unitary group, i.e. it takes in a unitary operator $U$ and produces the value of the final observable.  Under the composition control$\to$state$\to$observable (the so-called ``dynamical landscape''), the goal of optimal control is generally to maximize the observable with respect to the control.

The latter basic landscape decomposition has been applied to various quantum control objectives, including state-to-state transitions \cite{Rabitz2004, Rabitz2006, Hsieh2008}, general quantum mechanical observables on an ensemble \cite{Rabitz2006a, Ho2006, Wu2008}, and unitary transformation (quantum gate) preparation \cite{Rabitz2005, Hsieh2008a, Ho2009}.  This area of research has focussed on identifying and characterizing the critical points of the landscapes, revealing important features of the associated gradient flows.  In particular, this work has shown that the kinematic landscapes have no suboptimal extrema that can act as traps preventing the gradient flow from reaching a global optimum.

The present paper considers one aspect of the optimal control of a quantum mechanical observable on a $N$-level system.  Such observables can be defined for an ensemble of initial states (whose collective state is described by a density matrix $\rho$) by the kinematic landscape $J(U) = \Tr(U\rho U^{\dag}\Ob)$, where the observable is represented by a Hermitian matrix $\Ob$.  It is known that, in addition to a single global minimum submanifold and a single maximum submanifold, this kinematic landscape generally contains a large number of saddle submanifolds \cite{Wu2008}.  The number of saddles depends on the eigenstructure of $\rho$ and $\Ob$ and generally increases with the size $N$ of the quantum system.  %What is not well understood is the magnitude of the role of these saddles in the gradient flow of this landscape.  
By continuity, each saddle submanifold is surrounded by a neighborhood in which the norm of the gradient is small.  Any trajectory of the gradient flow entering such a neighborhood will be slowed by this small gradient and reach the global maximum less efficiently than trajectories that never enter such neighborhoods.  However, since the gradient flow is not attracted to saddles as it is to local maxima, the magnitude of the impact of these saddles on the overall gradient flow of this landscape is unclear.

With this perspective in mind, we derive estimates and bounds for the effective volume fractions of the critical submanifolds relative to the volume of the unitary group.  The term ``effective volume fraction'' is used here to mean the volume fraction of a region around a given critical submanifold in which the norm of the gradient is small, i.e. less than some $\epsilon >0$.  These volume fractions will serve as an indicator of the likelihood of a gradient integral curve falling under the influence of the saddles.  This influence can have a profound impact on the efficiency of a gradient ascent algorithm (whether implemented in the laboratory or in numerical simulations), or indeed any optimal control algorithm based on local information.% as the local flatness of the landscape in these near-critical regions can cause such an algorithm to ``stall out'', at least temporarily.

Additionally, we consider the asymptotic behavior of these volume fractions as the system size $N$ is allowed to increase.  By embedding the original $\rho$ and $\Ob$ in these larger spaces, this analysis addresses the role of truncation to finite $N$ in defining the characteristics of the gradient flow.  In particular, this research will show that the volume fractions tend to zero as $N\to\infty$, implying that the saddles should have a vanishingly small impact on the gradient flow.  This conclusion helps to explain the observed behavior of numerical quantum optimal control simulations, which show no evident increase in search effort as $N$ rises.

The paper is organized as follows.  Section \ref{sec:criticalSubmanifolds} reviews the structure of the critical submanifolds of $J(U) = \Tr(U\rho U^{\dag}\Ob)$.  The induced Hilbert-Schmidt measures of these submanifolds are computed in Section \ref{sec:criticalMeasure}.  The near-critical sets are described in Section \ref{sec:nearCriticalMeasure}, yielding estimates of the volumes of these sets in terms of a Haar measure.  Upper bounds for these volumes are developed in Section \ref{sec:asymptoticAnalysis} and the asymptotic behavior as $N\to\infty$ is considered.  The results are summarized Section \ref{sec:conclusions}.  Four appendices are included that detail the Chevalley lattice of direct products of Lie groups, the Hessian operator of $J$, the second fundamental forms of the critical submanifolds of $J$, and provide supporting arguments for a key conjecture.

\bigskip

\section{The Critical Submanifolds of the Landscape}
\label{sec:criticalSubmanifolds}
The landscape $J(U) = \Tr(U\rho U^{\dag}\Ob)$ is defined on the unitary group $\UN$, where the density matrix $\rho$ and the observable $\Ob$ are both $N\times N$ Hermitian matrices.  Without loss of generality, we will make the simplifying assumption that $\rho$ and $\Ob$ are both diagonal, with monotonically descending diagonal elements.  This may be done since passage from the landscape based on $\rho$ and $\Ob$ to the one based on their diagonalized forms represents merely a translation over the unitary group.  Since the metrics considered in this paper are bi-invariant, this translation may be neglected.%Now, it was shown in \cite{Wu2008} that the critical submanifolds of $J$ are of the form $\Orb(\perm):=\{V\perm W^{\dag}\;:\; (V,W)\in\Umn\}$, where $\perm\in S_{N}$ is any permutation matrix, $\Un = \mathrm{U}(n_{1})\oplus\cdots\oplus \mathrm{U}(n_{r})$, $\Um = \mathrm{U}(m_{1})\oplus\cdots\oplus\mathrm{U}(m_{s})$, and $\{n_{i}\}$ and $\{m_{j}\}$ are the multiplicities of the unique eigenvalues of $\rho$ and $\Ob$, respectively.

It will be fruitful at certain points in the analysis that follows to consider $\UN$ (with the geometry induced by the real Hilbert-Schmidt inner product on $\mathbb{C}^{N\times N}$) as a Riemannian globally symmetric space.  This will be done by fixing a permutation matrix $\perm\in S_{N}$ and letting $\cG:= \UN\otimes\UN$ and $\cK:=\{(V,\perm^{\dag}V\perm)\;:\;V\in\UN\}$.  Then, under the involutive analytic automorphism $\sigma(V,W) = (\perm W\perm^{\dag}, \perm^{\dag}V\perm)$, $(\cG,\cK)$ is a symmetric pair and $\cG/\cK$ is a Riemannian globally symmetric space under any $\cG$-invariant metric.  The map $\cG/\cK\to\UN$ given by $(V,W)\cK\mapsto V\perm W^{\dag}$ is clearly smooth, as is the section $U\mapsto (\mathbb{I}, U^{\dag}\perm)$.  Since the projection $\cG\to\cG/\cK$ is a smooth submersion, $V\perm W^{\dag}\mapsto(\mathbb{I},W\perm^{\dag} V^{\dag}\perm)\mapsto (V,W)\cK$ is a smooth inverse to the map $(V,W)\cK\mapsto V\perm W^{\dag}$, so this map is a diffeomorphism between $\cG/\cK$ and $\UN$.  The group action of $\cG$ on $\UN$ induced by the diffeomorphism is then given by $(V,W)\cdot U = V U W^{\dag}$, and the real Hilbert-Schmidt metric on $\UN$ is $\cG$-invariant, so $\UN$ is a Riemannian globally symmetric space under this structure.  The projection $\Phi_{\perm}:\cG\to\UN\simeq \cG/\cK$ can also be viewed as defining $\cG$ as a fiber bundle with base $\UN$ and fiber $\cK$.

%Let $\Phi: \Umn \times \UN\rightarrow \UN$ be the group action defined by $\Phi_{U}\big((V,W)\big) := \Phi\big((V,W), U\big) = V U W^{\dag}$.  %For a given permutation $\perm\in S_{N}\subset \UN$, let $\Phi_{\perm}:\Umn \rightarrow \UN$ be defined by $\Phi_{\perm}(V,W) = V^{\dag}\perm W$.  
Now, let $\cH := \Umn\subset\cG$, where $\Um = \mathrm{U}(m_{1})\oplus\cdots\oplus\mathrm{U}(m_{s})$, $\Un = \mathrm{U}(n_{1})\oplus\cdots\oplus\mathrm{U}(n_{r})$, and $\{m_{j}\}$ and $\{n_{i}\}$ are the multiplicities of the unique eigenvalues of $\Ob$ and $\rho$, respectively.  Then the critical set of $J(U) = \Tr(U\rho U^{\dag}\mathcal{O})$ has been shown to comprise a disjoint union of submanifolds of $\UN$, each of the form $\Phi_{\perm}(\cH)$ \cite{Wu2008}, where  $\Phi_{\perm}(\cH) = \Um\perm\Un$ is the orbit of the point $\perm\in S_{N}\subset\UN$ with respect to the induced action of $\cH$ on $\UN$.  So $\Phi_{\perm}(\cH) \cong \cH/\Stab_{\cH}(\perm)$, where the stabilizer $\Stab_{\cH}(\perm)$ is the subgroup of $\cH$ given by \begin{equation}\Stab_{\cH}(\perm) = \cH\cap\cK = \{(V,W)\in\Umn\;:\; \Phi_{\perm}(V,W) = V\perm W^{\dag} = \perm\}\end{equation} and is identified with $\UK\cong\Um\cap\perm\Un\perm^{\dag}$ in \cite{Wu2008}, i.e. $\Stab(\perm) = \{(\zeta, \perm^{\dag} \zeta \perm)\;:\; \zeta \in \Um\cap\perm\Un\perm^{\dag}\}$.  Here $\mathbf{K}$ is the $r \times s$ ``contingency table'' corresponding to $\mathbf{m}$, $\mathbf{n}$, and $\perm$.  For a given permutation $\perm$, $\Umn$ can be expressed as a fiber bundle \cite{Steenrod1951} as
	\begin{equation}\begin{CD} \Um\cap\perm\Un\perm^{\dag} @>{\zeta\mapsto(\zeta,\perm^{\dag}\zeta\perm)}>>  \Umn \\ @. @VV\Phi_{\perm}V\\  @. \qquad\Orb(\perm)\subset\UN.\end{CD}\nonumber\end{equation}
We seek to compute the ``volume'' (measure) of $\Orb(\perm)$ as an embedded submanifold of $\UN$, where $\UN$ is given the Riemannian metric induced by the real Hilbert-Schmidt inner product on $\mathbb{C}^{N\times N}$.  We will see that this problem reduces to one of computing the volumes of the Lie groups $\UK$ and $\Umn$ with respect to certain specific geometries.

%%Define the two spaces
%%\begin{subequations}
%%\begin{align}
%%	\mathfrak{k} & := \{(X,Y)\in \uN\otimes \uN \;:\; \rmd_{\mathbb{I}}\sigma(X,Y) = (X,Y)\} = \{(X,Y)\in\uN\otimes\uN\;:\; Y = \perm^{\dag}X\perm\}\\
%%	\mathfrak{p} & := \{(X,Y)\in \uN\otimes \uN \;:\; \rmd_{\mathbb{I}}\sigma(X,Y) = -(X,Y)\} = \{(X,Y)\in\uN\otimes\uN\;:\; Y = -\perm^{\dag}X\perm\}.
%%\end{align}
%%\end{subequations}
%%Then $\mathfrak{k}$ is the Lie algebra of $\cK$, the stabilizer of $\perm\in\UN$ with respect to the conjugation action of $\cG$ on $\UN$, and $\mathfrak{p}$ is the orthogonal complement of $\mathfrak{k}$.  The subspace $\mathfrak{p}$ is identified with the tangent space of $\UN$ at $\perm$.  Now let
%%\begin{equation}\mathfrak{s} := {(X,Y)\in\umn\;:\;

\begin{lemma}For any permutation matrix $\hat{\perm}\in S_{N}$, there exists a $\perm\in S_{N}\cap\Orb(\hat{\perm})$ such that $\Um\cap\perm\Un\perm^{\dag} = \UK$ in terms of block-diagonal structure and $\Un\cap\perm^{\dag}\Um\perm = \mathrm{U}(\mathbf{\tilde{K}}) = \perm^{\dag}\UK\perm$ where $\mathrm{U}(\mathbf{\tilde{K}})$ has essentially the same block structure as $\UK$, but the blocks are reordered.
\end{lemma}
\begin{proof}
For any $\Sigma\in S_{N}\cap\Um$ and $\Gamma\in S_{N}\cap\Un$, $\perm = \Sigma\hat{\perm}\Gamma^{\dag}\in S_{N}\cap\Orb(\perm)$.  Then $\Um\cap\perm\Un\perm^{\dag} = \Sigma^{\dag}\big(\Um\cap\hat{\perm}\Un\hat{\perm}^{\dag}\big)\Sigma$, so that $\Sigma$ may be chosen to reorder the rows and columns of each diagonal block of $\Um$ to group the matrix elements of $\Um\cap\hat{\perm}\Un\hat{\perm}^{\dag}$ according to the block of $\Un$ from whence they came, creating sub-blocks within each block of $\Um$.  This is exactly the structure desired for $\UK$.  Similarly, $\Un\cap\perm^{\dag}\Um\perm = \Gamma\big(\Un\cap\hat{\perm}^{\dag}\Um\hat{\perm}\big)$, so that $\Gamma$ may be likewise chosen to reorder the elements in each diagonal block of $\Un$ according to the originating block of $\Um$, constructing $\mathrm{U}(\mathbf{\tilde{K}})$.
\end{proof}

Without loss of generality, the remainder of the paper will use the permutation $\perm$ described in the above lemma, so that $\Um\cap\perm\Un\perm^{\dag} = \UK$.

\bigskip

\section{The Hilbert-Schmidt Measure of the Critical Submanifolds}
\label{sec:criticalMeasure}
We now turn to the problem of computing the volumes of the critical submanifolds.  To that end, we will choose geometries for $\UK$ and $\Umn$ such that there is a local isometry between $\Umn$ and $\Orb(\perm)\oplus\UK$.  Then it will be shown that the volume of $\Orb(\perm)$ is just the quotient of the volumes of $\Umn$ and $\UK$ under these specific geometries.

Let $\cA, \cB, \cC$ denote the following subspaces of $\umn$:
\begin{subequations}
\begin{align}
	\cA & = \{(X, \perm^{\dag}X\perm)\;:\; X\in\uK\}\\
	\cB & = \{(X, -\perm^{\dag}X\perm)\;:\; X\in\uK\}\\
	\cC & = (\cA\oplus\cB)^{\perp} = \{(Y, Z)\;:\; Y\in\um/\uK \text{ and } Z\in\un/\mathrm{u}(\tilde{\mathbf{K}})\}.
\end{align}
\end{subequations}
It may be readily verified that these spaces are mutually orthogonal in the Hilbert-Schmidt geometry and span $\umn$.  We extend them by left translation to form mutually orthogonal distributions over $\Umn$ that span the tangent space at each point. 
 
\begin{lemma}The distributions $\cA$, $\cB$, and $\cC$ are right invariant with respect to $\Stab(\perm)$.\end{lemma}
\begin{proof}
First observe that for any $\zeta\in \UK$% or $\tilde{\zeta}\in\mathrm{U}(\tilde{\mathbf{K}})$
, $X\in\uK$ if and only if $\zeta X\zeta^{\dag} \in \uK$.  Likewise $Y\in\um/\uK$ if and only if $\langle Y, X\rangle_{\rmHS} = 0$ for all $X\in \uK$, if and only if $\langle \zeta Y \zeta^{\dag}, \zeta X \zeta^{\dag}\rangle_{\rmHS} = 0$ for all $X\in\uK$, if and only if $\zeta Y\zeta^{\dag} \in \um/\UK$ and similarly for $Z\in\un/\mathrm{u}(\tilde{\mathbf{K}})$ with respect to $\tilde{\zeta}$.  In other words, the subspaces of $\umn$ given by evaluating the distributions $\cA$, $\cB$, and $\cC$ at the identity are invariant under the adjoint action of $\Stab(\perm)$.  Then, for any $(\zeta, \perm^{\dag}\zeta\perm) \in\Stab(\perm)$ and $(V,W)\in\Umn$
	\begin{subequations}
	\begin{align}\cA_{(V,W)\cdot(\zeta, \perm^{\dag}\zeta\perm)} & = \{(V,W)\cdot(\zeta, \perm^{\dag}\zeta\perm)\cdot(X,\perm^{\dag}X\perm\;:\;X\in\uK\}\\
%	& = \{(V,W)\cdot(\zeta X \zeta^{\dag}, \perm^{\dag}\zeta X \zeta^{\dag}\perm)\cdot(\zeta, \perm^{\dag}\zeta\perm \;:\; X\in\uK\}\\
	& = \{(V,W)\cdot(X, \perm^{\dag} X \perm)\cdot(\zeta, \perm^{\dag}\zeta\perm \;:\; X\in\uK\},
	\end{align}
	\end{subequations}
which is the right translation of $\cA_{(V,W)}$ to $(V,W)\cdot(\zeta, \perm^{\dag}\zeta\perm)$, so that $\cA$ is right invariant with respect to $\Stab(\perm)$.  Similar arguments apply to $\cB$ and $\cC$.
\end{proof}

Now, denote by $\hat{\Phi}_{\perm}$ the restriction of $\Phi_{\perm}$ to $\cH = \Umn$, and consider the images of these three distributions $\cA$, $\cB$, and $\cC$ under $\rmd \hat{\Phi}_{\perm}$.  Since for any $(V,W)\in\Umn$ and any $(\delta V, \delta W) \in \rmT_{(V,W)}\Umn$, there exists $(X,Y)\in\umn$ such that $(\delta V, \delta W) = (VX, WY)$, \begin{equation}\rmd_{(V,W)}\hat{\Phi}_{\perm}(\delta V, \delta W) = \rmd_{(V,W)}\hat{\Phi}_{\perm}(VX, WY) = VX\perm W^{\dag} - V\perm Y W^{\dag}.\end{equation}
Then the kernel of $\rmd_{(V,W)}\hat{\Phi}_{\perm}$ is $\{(VX,WY)\in\rmT_{(V,W)}\Umn\;:\; Y = \perm^{\dag}X\perm\}$, which is exactly the left-invariant distribution $\cA$ evaluated at $(V,W)$.  When acting on an element $(V,W)\cdot(X,-\perm^{\dag}X\perm)\in\cB_{(V,W)}$, $\rmd_{(V,W)}\hat{\Phi}_{\perm}$ yields $2VX\perm W^{\dag}$,  and when acting on an element $(V,W)\cdot(Y,Z)\in\cC_{(V,W)}$, $\rmd_{(V,W)}\hat{\Phi}_{\perm}$ yields $VY\perm W^{\dag} - V\perm Z W^{\dag}$.  Then, since $X\in \uK$ and $Y-\perm Z\perm^{\dag}\in\uN/\uK$, the $\rmd_{(V,W)} \hat{\Phi}_{\perm}$ images of $\cB_{(V,W)}$ and $\cC_{(V,W)}$ are orthogonal complements under any bi-invariant metric on $\UN$, in particular the Hilbert-Schmidt metric.  Furthermore, for $X_{1},X_{2}\in\uK$, $Y_{1},Y_{2}\in\um/\uK$, and $Z_{1},Z_{2}\in\un/\mathrm{u}(\tilde{\mathbf{K}})$, 
	\begin{subequations}
	\begin{align}
		\langle\rmd_{\mathbb{I}} \hat{\Phi}_{\perm}(X_{1},\;-\perm^{\dag}X_{1}\perm),\rmd_{\mathbb{I}} \hat{\Phi}_{\perm}(X_{2},-\perm^{\dag}X_{2}\perm)\rangle_{\rmHS} & = 2\langle (X_{1},-\perm^{\dag}X_{1}\perm), (X_{2},-\perm^{\dag}X_{2}\perm)\rangle_{\rmHS}\\
%		\langle\rmd_{(V,W)} \hat{\Phi}_{\perm}(X_{1},\;-\perm^{\dag}X_{1}\perm),\rmd_{(V,W)} \hat{\Phi}_{\perm}(X_{2},-\perm^{\dag}X_{2}\perm)\rangle_{\rmHS} & = 4\langle X_{1},X_{2}\rangle_{\rmHS}\\
%		& = 2\langle (X_{1},-\perm^{\dag}X_{1}\perm), (X_{2},-\perm^{\dag}X_{2}\perm)\rangle_{\rmHS}\\
		\langle\rmd_{\mathbb{I}} \hat{\Phi}_{\perm}(Y_{1},\;Z_{1}),\rmd_{\mathbb{I}} \hat{\Phi}_{\perm}(Y_{2},Z_{2})\rangle_{\rmHS} & = \langle (Y_{1},Z_{1}),\;(Y_{2},Z_{2})\rangle_{\rmHS}
%		\langle\rmd_{(V,W)} \hat{\Phi}_{\perm}(Y_{1},\;Z_{1}),\rmd_{(V,W)} \hat{\Phi}_{\perm}(Y_{2},Z_{2})\rangle_{\rmHS} & = \langle Y_{1},Y_{2}\rangle_{\rmHS} + \langle Z_{1},Z_{2}\rangle_{\rmHS}\\
%		& = \langle (Y_{1},Z_{1}),\;(Y_{2},Z_{2})\rangle_{\rmHS}
	\end{align}
	\end{subequations}
since $\perm Z_{1}\perm^{\dag},\perm Z_{2}\perm^{\dag}\in (\perm\un\perm^{\dag})/\uK\subset \uN/\um$ are orthogonal to $Y_{1}$ and $Y_{2}$ in the Hilbert-Schmidt metric.  So the image through $\rmd\hat{\Phi}_{\perm}$ of $\cA$ is zero, the restriction of $\rmd\hat{\Phi}_{\perm}$ to $\cB$ is two times a linear isometry, and the restriction to  $\cC$ is a linear isometry.

 Since $\hat{\Phi}_{\perm}$ is a fiber bundle with base $\Orb(\perm)$ and fiber $\UK$, for any $U_{0}\in\Orb(\perm)$, there exists a neighborhood $Q$ of $U_{0}$ and a smooth local section $\gamma_{Q}:Q\rightarrow \Umn$ such that $\hat{\Phi}_{\perm}(\gamma_{Q}(U)) = U$ for all $U\in Q$ and such that $\Image\big(\rmd_{U}\gamma_{Q}\big) = \cB_{\gamma_{Q}(U)}\oplus\cC_{\gamma_{Q}(U)}$, the orthogonal complement of $\cA_{\gamma_{Q}(U)} = \ker\big(\rmd_{\gamma_{Q}(U)}\hat{\Phi}_{\perm}\big)$ in $\rmT_{\gamma_{Q}(U)}\Umn$ with respect to the Hilbert-Schmidt metric.  Let $\Psi_{Q}:Q\oplus \UK \rightarrow \hat{\Phi}_{\perm}^{-1}(Q)\subset\Umn$ be defined by \begin{equation}\Psi_{Q}(U,\zeta):= \gamma_{Q}(U)\cdot(\zeta,\perm^{\dag}\zeta\perm).\end{equation}
Then \begin{equation}\rmd_{(U,\zeta)}\Psi_{Q}(\delta U, \delta \zeta) = \gamma_{Q}(U)\cdot (\delta\zeta, \perm^{\dag}\delta\zeta\perm) + \rmd_{U}\gamma_{Q}(\delta U)\cdot(\zeta, \perm^{\dag}\zeta\perm)\end{equation} where it may be observed that $\gamma_{Q}(U)\cdot(\delta\zeta, \perm^{\dag}\delta\zeta\perm)\in \cA_{\Psi_{Q}(U,\zeta)}$ by left invariance and the fact that $(\delta\zeta, \perm^{\dag}\delta\zeta\perm)\in\cA_{(\zeta, \perm^{\dag}\zeta\perm)}$ and also $\rmd_{U}\gamma_{Q}(\delta U)\cdot(\zeta, \perm^{\dag}\zeta\perm)\in \cB_{\Psi_{Q}(U,\zeta)}\oplus\cC_{\Psi_{Q}(U,\zeta)}$ by right invariance and the definition of $\gamma_{Q}$ whereby $d_{U}\gamma_{Q}(\delta U)\in \cB_{\gamma_{Q}(U)}\oplus\cC_{\gamma_{Q}(U)}$.

For any given $\perm$, we now define a Riemannian metric $\langle\cdot, \cdot\rangle_{\perm}$ on $\Umn$ at the point $\Upsilon = (V,W)$ as follows. For any $\delta\Upsilon_{1}, \delta\Upsilon_{2}\in\rmT_{\Upsilon}\big(\Umn\big)$, let
	\begin{subequations}
	\begin{align}
		\langle \delta\Upsilon_{1}, \delta \Upsilon_{2}\rangle_{\perm} :&= \frac{1}{2}\big\langle \delta \Upsilon_{1}^{\cA}, \delta \Upsilon_{2}^{\cA}\big\rangle_{\rmHS} + \big\langle \rmd_{\Upsilon}\hat{\Phi}_{\perm}(\delta\Upsilon_{1}^{\cB} + \delta\Upsilon_{1}^{\cC}),\rmd_{\Upsilon}\hat{\Phi}_{\perm}(\delta\Upsilon_{2}^{\cB} + \delta\Upsilon_{2}^{\cC})\big\rangle_{\rmHS}\\
		& = \frac{1}{2}\big\langle \delta \Upsilon_{1}^{\cA}, \delta \Upsilon_{2}^{\cA}\big\rangle_{\rmHS} + 2\big\langle \delta\Upsilon_{1}^{\cB}, \delta\Upsilon_{2}^{\cB}\big\rangle_{\rmHS} + \big\langle \delta\Upsilon_{1}^{\cC}, \delta\Upsilon_{2}^{\cC}\big\rangle_{\rmHS},
	\end{align}
	\label{eqn:PiMetric}
	\end{subequations}
where we have used the orthogonal decomposition $\delta\Upsilon = \delta\Upsilon^{\cA} + \delta\Upsilon^{\cB} + \delta\Upsilon^{\cC}$ of $\delta \Upsilon$ into the three orthogonal subspaces given by $\cA$, $\cB$, and $\cC$. Then for any $U\in Q$ and $\zeta\in \UK$, 
	\begin{subequations}
	\begin{align}
		\lefteqn{\big\langle \rmd_{(U,\zeta)}\Psi_{Q}(\delta U_{1}, \delta\zeta_{1}),\; \rmd_{(U,\zeta)}\Psi_{Q}(\delta U_{2}, \delta\zeta_{2})\big\rangle_{\perm}}\nonumber\\
		& =\frac{1}{2}\big\langle \gamma_{Q}(U)\cdot(\delta\zeta_{1}, \perm^{\dag}\delta\zeta_{1}\perm), \; \gamma_{Q}(U)\cdot(\delta\zeta_{2}, \perm^{\dag}\delta\zeta_{2}\perm)\rangle_{\rmHS}\nonumber\\
		& \qquad + \big\langle \rmd_{\Psi_{Q}(U,\zeta)}\hat{\Phi}_{\perm}\big(\rmd_{U}\gamma_{Q}(\delta U_{1})\cdot(\zeta, \perm^{\dag}\zeta\perm)\big), \;\rmd_{\Psi_{Q}(U,\zeta)}\hat{\Phi}_{\perm}\big(\rmd_{U}\gamma_{Q}(\delta U_{2})\cdot(\zeta, \perm^{\dag}\zeta\perm)\big)\big\rangle_{\rmHS}\\
		& = \big\langle \delta\zeta_{1}, \delta\zeta_{2}\rangle_{\rmHS} + \big\langle \rmd_{\gamma_{Q}(U)}\hat{\Phi}_{\perm}\big(\rmd_{U}\gamma_{Q}(\delta U_{1})\big), \;\rmd_{\gamma_{Q}(U)}\hat{\Phi}_{\perm}\big(\rmd_{U}\gamma_{Q}(\delta U_{2})\big)\big\rangle_{\rmHS}\\
%		& = \big\langle \delta\zeta_{1}, \delta\zeta_{2}\rangle_{\rmHS} + \big\langle \rmd_{\gamma_{Q}(U)}\hat{\Phi}_{\perm}\big(\rmd_{U}\gamma_{Q}(\delta U_{1})\big), \;\rmd_{\gamma_{Q}(U)}\hat{\Phi}_{\perm}\big(\rmd_{U}\gamma_{Q}(\delta U_{2})\big)\big\rangle_{\rmHS}\\
		& = \big\langle \delta\zeta_{1}, \delta\zeta_{2}\rangle_{\rmHS} + \big\langle\delta U_{1}, \delta U_{2}\big\rangle_{\rmHS},
	\end{align}
	\end{subequations}
where the last step follows from the fact that $\hat{\Phi}_{\perm}\circ\gamma_{Q}$ is the identity map on $Q$.  So, if we endow $\UK\subset \mathbb{C}^{N\times N}$ with the Riemannian metric induced by the real Hilbert-Schmidt inner product on $\mathbb{C}^{N\times N}$, then $\Psi_{Q}$ is an \emph{isometric} diffeomorphism between $Q\oplus\UK$ and $\hat{\Phi}_{\perm}^{-1}(Q)\subset \Umn$.  It then follows that $\Vol_{\rmHS}(Q)\Vol_{\rmHS}\big(\UK\big) = \Vol_{\perm}\big(\hat{\Phi}_{\perm}^{-1}(Q)\big)$, and therefore \begin{equation}\Vol_{\rmHS}\big(\Orb(\perm)\big) = \frac{\Vol_{\perm}\big(\Umn\big)}{\Vol_{\rmHS}\big(\UK\big)}\label{eqn:pvolQuotientPerm}\end{equation}
so the problem reduces to computing the ratio of the volumes of the two Lie groups with respect to the indicated geometries.

The expression in \eqref{eqn:pvolQuotientPerm} may be simplified further by considering the volume form on $\Umn$ induced by the metric $\langle\cdot,\cdot\rangle_{\perm}$ \cite{Warner1983}.  This volume form can be realized by choosing orthonormal bases of $\cA$, $\cB$, and $\cC$ at $(\mathbb{I},\mathbb{I})$ in the $\langle\cdot,\cdot\rangle_{\perm}$ metric and extending them to orthonormal vector fields by left translation.  Denote these fields by $\{F_{i}\}$ where $i = 1,\dots, d$ where $d = \sum m_{j}^{2} + \sum n_{i}^{2}$.  Then construct the dual basis of 1-forms $\omega_{i} = \langle F_{i}, \cdot\rangle_{\perm}$ and the volume form by $\omega_{i}\wedge\cdots\wedge\omega_{d}$.  Let $\kappa = \sum k_{ij}^{2} = \|\mathbf{K}\|_{\rmHS}^{2}$.  Because of the relationship between $\langle \cdot, \cdot\rangle_{\perm}$ and the Hilbert-Schmidt metric described in \eqref{eqn:PiMetric}, if $F_{1},\dots,F_{\kappa}$ is the basis for $\cA$ and $F_{\kappa + 1}, \dots, F_{2\kappa}$ is the basis for $\cB$ under $\langle \cdot, \cdot\rangle_{\perm}$, then $\frac{1}{\sqrt{2}}F_{1},\dots,\frac{1}{\sqrt{2}}F_{\kappa}$ and $\sqrt{2}F_{\kappa + 1}, \dots, \sqrt{2}F_{2\kappa}$ are the corresponding orthonormal vector fields under the Hilbert-Schmidt metric.  So $\sqrt{2}\omega_{1},\dots, \sqrt{2}\omega_{\kappa}$ and $\frac{1}{\sqrt{2}}\omega_{\kappa+1},\dots, \frac{1}{\sqrt{2}}\omega_{2\kappa}$ are the corresponding 1-forms under the Hilbert-Schmidt metric.  As a result, the volume form $\hat{\omega}$ on $\Umn$ induced by the Hilbert-Schmidt metric is identical to $\omega$, the volume form induced by $\langle \cdot, \cdot\rangle_{\perm}$:
\begin{equation}\hat{\omega} = \sqrt{2}\omega_{1}\wedge\cdots\wedge\sqrt{2}\omega_{\kappa}\wedge\frac{1}{\sqrt{2}}\omega_{\kappa +1}\wedge\cdots\wedge\frac{1}{\sqrt{2}}\omega_{2\kappa}\wedge\omega_{2\kappa+1}\wedge\cdots\wedge\omega_{d} = \omega.\end{equation}
So, while $\langle\cdot,\cdot\rangle_{\perm}$ defines a different geometry on $\Umn$ compared to the Hilbert-Schmidt metric, stretching some dimensions and shrinking others, the result is no net difference in the volume form and therefore no difference in the volume under these two geometries.  So we may replace \eqref{eqn:pvolQuotient} with the expression \begin{equation}\Vol_{\rmHS}\big(\Orb(\perm)\big) = \frac{\Vol_{\rmHS}\big(\Umn\big)}{\Vol_{\rmHS}\big(\UK\big)},\label{eqn:pvolQuotient}\end{equation} which reduces the problem to one of computing the volumes of $\Umn$ and $\UK$, both just direct sums of unitary groups, under the Hilbert-Schmidt metric.

\begin{lemma}[The Volume of $\mathrm{U}(\mathbf{a})$]
Let $\mathbf{a} \in \mathbb{N}^{b}$ be any vector of non-negative integers.  Then the volume of $\mathrm{U}(\mathbf{a}) = \mathrm{U}(a_{1})\oplus\cdots\oplus\mathrm{U}(a_{b})$ with respect to the Riemannian metric induced by the real Hilbert-Schmidt inner product is given by
\begin{equation}
\Vol_{\rmHS}\big(\mathrm{U}(\mathbf{a})\big) = \frac{(2\pi)^{\frac{1}{2}\sum a_{l}^{2} + \frac{\bar{a}}{2}}}{\prod_{l}\prod_{s=0}^{a_{l}-1}s!}.
\end{equation}
\end{lemma}
\begin{proof}To apply Macdonald's formula for the volume of a Lie group \cite{Macdonald1980, Hashimoto1997}, we need the basis for the Chevalley lattice given by the vectors $\{\tau_{j}\}$, $\{\xi_{\alpha}\}$, and $\{\eta_{\alpha}\}$ in Appendix \ref{sec:ChevalleyLattice}.  These vectors are all mutually orthogonal in the real Hilbert-Schmidt inner product on $\mathrm{U}(\mathbf{a})$.  In addition, the $\tau_{j}$'s have norm $1$, and the $\xi_{\alpha}$'s and $\eta_{\alpha}$'s have norm $\sqrt{2}$.  Letting $\bar{a} = \sum a_{l}$, the Gram matrix of this basis is then a diagonal matrix with $\bar{a}$ entries equal to one and $\sum_{l}a_{l}(a_{l}-1)$ entries equal to two, so that the volume of the fundamental cell is just \begin{equation}\lambda(\mathfrak{g}/\gZ) = 2^{\frac{1}{2}\sum_{l}a_{l}(a_{l}-1)} = 2^{\frac{1}{2}\sum_{l}a_{l}^{2}-\frac{\bar{a}}{2}}.\end{equation}
Then Macdonald's formula gives the volume of $\UK$ as 
	\begin{subequations}
	\begin{align}
		\Vol_{\rmHS}\big(\mathrm{U}(\mathbf{a})\big) & = 2^{\frac{1}{2}\sum_{l}a_{l}^{2}-\frac{\bar{a}}{2}}\prod_{l}\prod_{s=0}^{a_{l}-1}\Vol\big(S^{2s+1}\big)\\
%		& = \prod_{l}\frac{(2\pi)^{\frac{a_{l}(a_{l}+1)}{2}}}{\prod_{s=0}^{a_{l}-1}s!}\\
		& = \frac{(2\pi)^{\frac{1}{2}\sum a_{l}^{2} + \frac{\bar{a}}{2}}}{\prod_{l}\prod_{s=0}^{a_{l}-1}s!}\\
		& = \prod_{l}\Vol_{\rmHS}\big(\mathrm{U}(a_{l})\big).
	\end{align}
	\label{eqn:volUnitaryProd}
	\end{subequations}
\end{proof}

With this lemma, we can now compute the volume of the critical submanifold $\Orb(\perm)$:
	\begin{subequations}
	\begin{align}
		\Vol_{\rmHS}\big(\Orb(\perm)\big) & = \frac{\Vol_{\rmHS}(\Umn)}{\Vol_{\rmHS}\big(\UK\big)} = \frac{\prod\Vol_{\rmHS}\big(\mathrm{U}(n_{i})\big)\prod\Vol_{\rmHS}\big(\mathrm{U}(m_{j})\big)}{\prod_{ij}\Vol_{\rmHS}\big(\mathrm{U}(k_{ij})\big)}\\
		& = \frac{(2\pi)^{\frac{1}{2}\sum n_{i}^{2} + \frac{1}{2}\sum m_{j}^{2} + N}}{(2\pi)^{\frac{1}{2}\sum k_{ij}^{2} + \frac{N}{2}}}\frac{\prod_{ij}\prod_{r=0}^{k_{ij}-1}r!}{\prod_{i}\prod_{p=0}^{n_{i}-1}p!\prod_{j}\prod_{q=0}^{m_{j}-1}q!}\\
		 & = (2\pi)^{\frac{d+N}{2}}\frac{\prod_{ij}\prod_{r=0}^{k_{ij}-1}r!}{\prod_{i}\prod_{p=0}^{n_{i}-1}p!\prod_{j}\prod_{q=0}^{m_{j}-1}q!},
	\end{align}
	\end{subequations}
where $d = \sum m_{j}^{2} + \sum n_{i}^{2} - \sum k_{ij}^{2}$ is the dimension of $\Orb(\perm)$.

\bigskip

\subsection{Examples}
\label{sec:volExamples}
\begin{example}[Maximum Submanifold of $P_{i\to f}$]
	For any two (non-zero) vectors $|i\rangle$ and $|f\rangle$ in $\mathbb{C}^{N}$, let $\rho = \frac{|i\rangle\langle i|}{\langle i|i\rangle}$ and $\Ob = \frac{|f\rangle\langle f|}{\langle f|f\rangle}$.  Then $J(U) = \Tr(U\rho U^{\dag}\Ob)$ represents the transition probability from initial state $|i\rangle$ to final state $|f\rangle$.  Translating the problem by diagonalizing $\rho$ and $\Ob$ and sorting the eigenvalues in decreasing order, it is found that $\rho = \Ob$ has a single "1" in the (1,1) element and zero elsewhere.  The $N^{2}-2N+2$ dimensional maximum submanifold then corresponds to the identity permutation and yields the following contingency table \cite{Wu2008} and volume:
	\begin{equation}\begin{array}{l|l|l|}
			& m_{1} = 1 & m_{2} = N-1 \\
			\hline
			n_{1} = 1 & k_{11} = 1 & k_{12} = 0\\
			\hline
			n_{2} = N-1 & k_{21} = 0 & k_{22} = N-1\\
			\hline
		\end{array} \qquad\qquad \Vol_{\rmHS}\Orb(\mathbb{I}) = \frac{(2\pi)^{\frac{1}{2}(N^{2}-N+2)}}{\prod_{p=0}^{N-2}p!}.
	\end{equation}
\end{example}
\begin{example}[Minimum Submanifold of $P_{i\to f}$]
	If, in the previous example, a permutation is used that fails to align the non-zero eigenvalues of $\rho$ and $\Ob$, then $\Orb(\perm)$ is the $N^{2}-2$ dimensional minimum submanifold, with contingency table and volume:
	\begin{equation}\begin{array}{l|l|l|}
			& m_{1} = 1 & m_{2} = N-1 \\
			\hline
			n_{1} = 1 & k_{11} = 0 & k_{12} = 1\\
			\hline
			n_{2} = N-1 & k_{21} = 1 & k_{22} = N-2\\
			\hline
		\end{array} \qquad\qquad \Vol_{\rmHS}\Orb(\perm) = \frac{(2\pi)^{\frac{1}{2}(N^{2}+N-2)}}{(N-2)!\prod_{p=0}^{N-2}p!}.
	\end{equation}
\end{example}

\begin{example}[Fully Non-degenerate $\rho$ and $\Ob$]
	In the case where $\rho$ and $\Ob$ are both fully non-degenerate, $m_{1} = \dots = m_{N} = 1$, $n_{1} = \dots = n_{N} = 1$ and $\mathbf{K} = \perm^{\dag}$.  Then for any critical submanifold $\Orb(\perm)$, all of which are $N$ dimensional tori,
	\begin{equation}\Vol_{\rmHS}\Orb(\perm) = (2\pi)^{N}.
	\end{equation}
%	{\bf{So this is a rare case where the volume is increasing with dimension, although it requires that as N increases, $\rho$ and $\Ob$ are given new and different eigenvalues to maintain non-degeneracy.  Since the volume of $\UN$ is decreasing with $N$, the cartoon picture here might be of a string on a surface (sphere or torus, say) with the surface shrinking while the string gets longer and longer.  On the other hand, do the added dimensions yield more "room"?  Does this example admit any more detailed analysis?  The critical submanifolds are (maximal?) tori.}}
\end{example}

\bigskip

\section{The Measure of the Near-Critical Set}
\label{sec:nearCriticalMeasure}
In this section, an estimate is derived for the measure of the near-critical set $C_{\epsilon}^{\perm} = \{U \;:\; \|\grad J(U)\|\leq \epsilon\}\subset\UN$ surrounding the critical submanifold $\Orb(\perm)$. To do that, we first approximate the near-critical set by an ellipsoidal tube about $\Orb(\perm)$.
 
\subsection{Approximating the Set of Interest}
Let $U\in\Orb(\perm)$ be a critical point of $J$, and let $X\in \uN$ be such that $\|X\|= 1$ and $UX\in \big(\rmT_{U}\Crit(J)\big)^{\perp}$, i.e. $UX\in T_{U}\UN$ is orthogonal to the null space of $\Hess_{J,U}:\rmT_{U}\UN\to\rmT_{U}\UN$ (the Hessian operator of $J$ at the point $U\in\UN$) since $J$ is Morse-Bott, the necessary properties having been established in \cite{Wu2008}.  Define $F_{X}:\mathbb{R}\rightarrow \mathbb{R}$ by \begin{equation}F_{X}(s) = \big\|\grad J\big(U\exp(sX)\big)\big\|^{2}.\end{equation}  Then $F_{X}(0) = 0$ and \begin{equation}\frac{\rmd F_{X}}{\rmd s} = 2\Big\langle \grad J\big(U\exp(sX)\big),\; \Hess_{J, U\exp(sX)}\big(U\exp(sX)X\big)\Big\rangle\end{equation} so that $(\rmd F_{X}/\rmd s)(0) = 0$, since $U\in\Crit(J)$.  Furthermore, \begin{align}\frac{\rmd^{2} F_{X}}{\rmd s^{2}} & = 2\Big\langle \Hess_{J, U\exp(sX)}\big(U\exp(sX)X\big), \; \Hess_{J, U\exp(sX)}\big(U\exp(sX)X\big)\Big\rangle\nonumber\\ & \qquad + 2\Big\langle \grad J\big(U\exp(sX)\big), \; \nabla_{U\exp(sX)X}\Hess_{J, U\exp(sX)}\big(U\exp(sX)X\big)\Big\rangle\end{align} so that $(\rmd^{2} F_{X}/\rmd s^{2})(0) = 2\|\Hess_{J,U}(UX)\|^{2}$.  Thus, for small $s$, $F_{X}(s) = s^{2}\|\Hess_{J,U}(UX)\|^{2} + \mathscr{O}(s^{3})$.  Then in order to have $\big\|\grad J\big(U\exp(sX)\big)\big\|\leq \epsilon$, we should have $F_{X}(s)\leq \epsilon^{2}$, and therefore $s^{2}\leq \frac{\epsilon^{2}}{\|\Hess_{J,U}(UX)\|^{2}} + \mathscr{O}(\epsilon^{3})$.  

Now, suppose that $\{UY_{i}\}$ are the orthonormal eigenvectors of $\Hess_{J,U}$ corresponding to non-zero eigenvalues $\{\beta_{i}\}$ (see Appendix \ref{sec:HessianStructure}).  Then each normalized $UX\in\big(\rmT_{U}\Crit(J)\big)^{\perp}$ is such that $X$ can be written $X = \sum\alpha_{i}Y_{i}$ with $\sum \alpha_{i}^{2} = 1$.  Then the condition on $s$ is that $s^{2}\leq \frac{\epsilon^{2}}{\sum \alpha_{i}^{2}\beta_{i}^{2}}$, i.e. $\sum (s\alpha_{i})^{2}\beta_{i}^{2}\leq \epsilon^{2}$, so that $(s\alpha_{1},\dots,s\alpha_{m})$ is a point in the $m$-dimensional solid ellipsoid with principal axes $\{\epsilon/|\beta_{i}|\}$.  Therefore $sX$ lies in the $m$-dimensional solid ellipsoid with principal axes $\{(\epsilon/|\beta_{i}|)Y_{i}\}$, and the set of all $U\exp(sX)$ for $sX$ in this ellipsoid is an $m$-dimensional geodesic ellipsoid in $\UN$, which we will denote by $\mathfrak{E}_{\epsilon}(U)$.  Repeating this analysis at every point $U$ of the critical submanifold $\Orb(\perm)$ and drawing together the resulting geodesic ellipsoids yields an ellipsoidal tube $\mathfrak{T}_{\epsilon} = \cup_{U\in\Orb(\perm)}\mathfrak{E}_{\epsilon}(U)$ about $\Orb(\perm)$ that approximates the set of points for which $\|\grad J\|\leq \epsilon$.  It is this tube of near-critical points whose volume we will estimate.  Figure \ref{fig:tubeCartoon} offers one simple example of such an ellipsoidal tube about the submanifold $S^{1}$ of $\mathbb{R}^{3}$.
\begin{figure}
\includegraphics[clip = true]{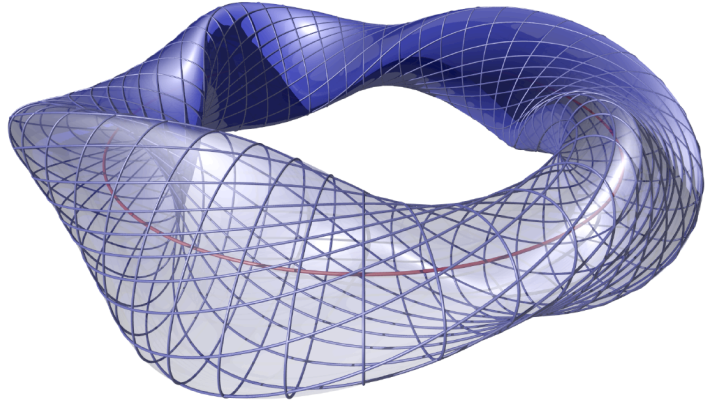}
\caption{Schematic rendering of a simple example of an ellipsoidal tube about a circular 1-dimensional submanifold of $\mathbb{R}^{3}$.  The ellipsoidal tubes considered in this paper are of a similar nature, but more complex, being tubes about higher dimensional (and more topologically and geometrically interesting) submanifolds of $\UN$.}
\label{fig:tubeCartoon}
\end{figure}

\bigskip

\subsection{Volumes of Tubes}
The study of the volumes of tubes goes back to 1939, when Hermann Weyl \cite{Weyl1939} gave the first complete description of the volumes of \emph{spherical} tubes about submanifolds of Euclidean and spherical spaces.  These volumes were presented as a finite power series in the radius of tube, with coefficients derived from the geometry of the submanfold.  This beautiful result, frequently referred to as Weyl's tube formula, has been successfully extended to certain other very special spaces (e.g. projective spaces), but appears not to have been extended to Lie groups or to the unitary group in particular.  However, we may still make use of an infinite power series \cite{Gray2004} to approximate the desired volume.

%%In \cite[Theorem 9.23]{Gray2004} we find the following expression for the volume of a spherical tube of radius $r$ about a $\smd$-dimensional submanifold $P$ of an $n$ dimensional manifold $M$:
%%\begin{equation}V_{P}^{M}(r) = \frac{(\pi r^{2})^{\frac{1}{2}(n-\smd)}}{\Gamma\big((n-\smd)/2 + 1\big)}\int_{P}\Big\{1 + Ar^{2} + Br^{4} + O\big(r^{6}\big)\Big\}\,\rmd P,\end{equation}
%%where $A$ and $B$ are derived from the the Riemannian curvature tensors of $P$ and $M$ and from the second fundamental form of the embedding of $P$ in $M$.  We will take advantage of this result by introducing a local change in geometry that converts our ellipsoidal tube into a spherical one.

%%\bigskip
%%{\bf{...clear, cogent, and well-reasoned arguments go here...}}
%%\bigskip

We will closely follow the notation, definitions, and conventions in \cite{Gray2004}.  Let $\exp_{\nu}:\nu\rightarrow \UN$ denote the Riemannian exponential map from the normal bundle $\nu$ of $\Orb(\perm)$ to $\UN$, i.e. $\exp_{\nu}$ takes a point $(U,V)\in \nu$ with $V\in \big(\rmT_{U}\Orb(\perm)\big)^{\perp}$, and outputs the point in $\UN$ found by following the constant speed geodesic $\xi$ with $\xi(0) = U$ and $\xi'(0) = V$ out to $\xi(1)$.  Let $y_{1},\dots,y_{\smd}$ be a local coordinate system for $\Orb(\perm)$.  Let $x_{1},\dots,x_{N^{2}}$ be the Fermi coordinates \cite{Gray2004} on $\UN$ generated by $y_{1},\dots, y_{\smd}$ and the orthonormal fields $E_{\smd+1},\dots, E_{N^{2}}$ on $\Orb(\perm)$ defined as the orthonormal eigenvectors of the Hessian of $J$ corresponding to non-zero eigenvalues, as in \eqref{eqn:eigenvectors},
\begin{equation}
	x_{i}\left(\exp_{\nu}\left(U, \, \sum_{j=\smd+1}^{N^{2}}t_{j}E_{j}(U)\right)\right) = \begin{cases}y_{i}(U) & 1\leq i\leq \smd \\ t_{i} & \smd+1\leq i\leq N^{2}\end{cases}\end{equation}Note that the vector fields $\{E_{i}\}$ track the principal axes of the ellipsoid normal to $\Orb(\perm)$.  For each $i\in\{1,\dots, N^{2}\}$, let $X_{i}$ denote the vector field $\frac{\partial}{\partial x_{i}}$.  The geodesic ellipsoidal tubular shell characterized by the ``radius'' $r$, denoted $\tilde{\mathfrak{E}}_{r}$, is described in this local coordinate chart by the set of points for which $\sum_{i=\smd+1}^{N^{2}}x_{i}^{2}\beta_{i}^{2} = r^{2}$.  Let $L\in\mathfrak{X}\big(\UN - \Orb(\perm)\big)$ denote the (outward) unit normal vector field to the shells $\tilde{\mathfrak{E}}_{r}$, let $\hat{L}$ denote the differential 1-form $\hat{L} = \langle L, \cdot \rangle$, and let $\mu$ denote the volume form of $\tilde{\mathfrak{E}}_{r}$.  Then $\hat{L}\wedge\mu = \omega$, where $\omega$ is the volume form on $\UN$.  Likewise let $L_{\nu}$ be the unit normal vector field to the $\epsilon$-ellipsoid in $\big(\rmT_{U}\Orb(\perm)\big)^{\perp}$ and let $\hat{L}_{\nu}$ be the 1-form $\hat{L}_{\nu} = \langle L_{\nu}, \cdot\rangle$.  Let $\mu_{\nu}$ denote the volume form of the ellipsoidal shell in $\big(\rmT_{U}\Orb(\perm)\big)$, so that $\hat{L}_{\nu}\wedge\mu_{\nu} = \omega_{\nu}$ where $\omega_{\nu}$ is the volume form on the normal bundle $\nu$.

We would like to compare $\exp_{\nu}^{*}(\hat{L})$, the pullback of $\hat{L}$ to $\nu$,  with $\hat{L}_{\nu}$.  To that end, observe that at a point $V = \sum_{i=\smd+1}^{N^{2}}t_{i}E_{i}$ in the normal space $\big(\rmT_{U}\Orb(\perm)\big)^{\perp}$, the tangent space to the ellipsoidal shell through $V$ is the set $\{\sum_{i=\smd+1}^{N^{2}}\tau_{i}E_{i}\}$ where $\sum_{i=\smd+1}^{N^{2}}\tau_{i}t_{i}\beta_{i}^{2} = 0$.  Then the unit normal is $L_{\nu}(U,V) = \frac{\sum_{i=\smd+1}^{N^{2}}\beta_{i}^{2}t_{i}E_{i}}{\sqrt{\sum \beta_{j}^{4}t_{j}^{2}}}\in \big(\rmT_{U}\Orb(\perm)\big)^{\perp}$.  At the corresponding point $W = \exp_{\nu}(U,V)\in\UN$, the tangent space to the geodesic ellipsoid is the set 
\begin{equation}\left\{\rmd\exp_{\nu}\left(\sum_{i=\smd+1}^{N^{2}}\tau_{i}E_{i}\right)\right\} = \left\{\sum_{i=\smd+1}^{N^{2}} \tau_{i}X_{i}\right\}\label{eqn:tangentspacegeodellipsoid}\end{equation}
where $\sum_{i=\smd+1}^{N^{2}}\tau_{i}t_{i}\beta_{i}^{2} = 0$.  Then $L(W)$ is the unit vector in $\rmT_{W}\mathrm{U}(N)$ normal to this space, i.e., $\sum_{i=\smd+1}^{N^{2}}\tau_{i}\langle L, X_{i}\rangle = 0$ for all $\{\tau_{i}\}$ such that $\sum_{i=\smd+1}^{N^{2}}\tau_{i}t_{i}\beta_{i}^{2} = 0$.  This implies that $\langle L, X_{i}\rangle = c t_{i}\beta_{i}^{2}$ for all $i\in\{\smd+1, \dots, N^{2}\}$ and some fixed normalization constant $c$.  Let $Z\in \big(\rmT_{U}\Orb(\perm)\big)^{\perp}$ be perpendicular to $L$ [i.e., tangent to the ellipsoid through $(U,V)$].  Then by \eqref{eqn:tangentspacegeodellipsoid} $\rmd \exp_{\nu}(Z)$ is tangent to the geodesic ellipsoid through $W$, so that $\langle (\rmd\exp_{\nu})^{*}(L), Z\rangle = \langle L, \rmd \exp_{\nu}(Z)\rangle = 0$, meaning that $(\rmd\exp_{\nu})^{*}(L) = c L_{\nu}$, where it can be shown that $c = \|(\rmd\exp_{\nu})^{*}(L)\| = 1 + O(\epsilon^{2})$.  Hence $\hat{L}_{\nu} = (1+O(\epsilon^{2}))\exp_{\nu}^{*}(\hat{L})$.  %Let $G$ be the Gram matrix $G_{ij} = \langle X_{i}, X_{j}\rangle$ for $i,j\in\{\smd+1, \dots, N^{2}\}$ and let $\Xi_{i} = \sum_{i=\smd+1}^{N^{2}}\big(G^{-\frac{1}{2}}\big)_{ij}X_{j}$.  Then $\{\Xi_{i}\}$ for $i\in\{\smd+1, \dots, N^{2}\}$ is an orthonormal frame with the same span as the corresponding $\{X_{i}\}$.  It may be shown that $\Xi_{i} = X_{i} + \sum_{j=\smd+1}^{N^{2}}O(\epsilon^{2})X_{j}$.  
%Now, since, $L$ is unit length, \begin{equation}1 = \langle L,L\rangle = \sum_{i=\smd+1}^{N^{2}}\langle L, \Xi_{i}\rangle^{2} = \sum_{i=\smd+1}^{N^{2}}\langle L, X_{i}\rangle^{2} + O(\epsilon^{2}) = c^{2}\sum_{i=\smd+1}^{N^{2}}\beta_{i}^{4}t_{i}^{2} + O(\epsilon^{2})\end{equation} from which it follows that $c = \frac{1}{\sqrt{\sum_{i=\smd+1}^{N^{2}}\beta_{i}^{4}t_{i}^{2}}} + O(\epsilon^{2})$.  So we find that $\langle L, X_{i}\rangle  = \frac{\beta_{i}^{2}t_{i}}{\sqrt{\sum_{j=\smd+1}^{N^{2}}\beta_{j}^{4}t_{j}^{2}}} + O(\epsilon^{2}) = \langle L_{\nu}, E_{i}\rangle  + O(\epsilon^{2})$.  Now we can see that 
%\begin{align}
%	\exp_{\nu}^{*}(\hat{L})(W) & = \hat{L}(\rmd\exp_{\nu}(W)) = \langle L, \rmd\exp_{\nu}(W)\rangle = \sum_{i=\smd+1}^{N^{2}}W_{i}\langle L, \rmd\exp_{\nu}(E_{i})\rangle  = \sum_{i=\smd+1}W_{i}\langle L, X_{i}\rangle + O(\epsilon^{2})\\
%	& = \sum W_{i}\langle L_{\nu}, E_{i}\rangle + O(\epsilon^{2}) = \langle L_{\nu}, W\rangle + O(\epsilon^{2}) = \hat{L}_{\nu}(W) + O(\epsilon^{2}).
%\end{align}
%In fact, one can show that \begin{equation}\exp_{\nu}^{*}(\hat{L}) = \hat{L}_{\nu} + \sum_{i=\smd+1}^{N^{2}}O(\epsilon^{2})\hat{E}_{i}.\end{equation}

With the above material in mind, and taking $\rmd V(U)$ to be the volume measure of the ellipsoid within $\big(\rmT_{U}\Orb(\perm)\big)^{\perp}$ and $\rmd P$ to be the volume measure of $\Orb(\perm)$, we find that the pull-back of $\omega$ is \cite{Gray2004} 
\begin{subequations}
\begin{align}
	\exp_{\nu}^{*}(\hat{L})\wedge\exp_{\nu}^{*}(\mu) & = \exp_{\nu}^{*}(\omega)(U,V) = \omega(X_{1},\dots,X_{N^{2}})\big(\exp_{\nu}(U,V)\big)\omega_{\nu}(U,V)\nonumber\\
	 & = \omega(X_{1},\dots,X_{N^{2}})\big(\exp_{\nu}(U,V)\big)L_{\nu}\wedge\mu_{\nu}(U,V)\\
	 & = \omega(X_{1},\dots,X_{N^{2}})\big(\exp_{\nu}(U,V)\big)L_{\nu}\wedge\rmd V\wedge\rmd P\\
	 & = \big(1+O(\epsilon^{2})\big)\omega(X_{1},\dots,X_{N^{2}})\big(\exp_{\nu}(U,V)\big)\exp_{\nu}^{*}(\hat{L})\wedge\rmd V\wedge\rmd P
\end{align}  
\end{subequations}
and therefore \begin{equation}\exp_{\nu}^{*}(\mu) = \big(1+O(\epsilon^{2})\big)\omega(X_{1},\dots,X_{N^{2}})\big(\exp_{\nu}(U,V)\big)\rmd V\wedge\rmd P.\end{equation}

Now, it was shown in \cite{Gray2004} that \begin{equation}\omega(X_{1},\dots,X_{N^{2}}) = 1 - \sum_{i=\smd+1}^{N^{2}}\langle H, X_{i}\rangle x_{i} + \text{higher order terms},\end{equation} where $H$ is a section of the normal bundle $\nu$ over $\Orb(\perm)$ called the mean curvature field.  The first order term in the above expression for $\omega$ will integrate to zero (as will all odd order terms) due to the symmetry of the ellipse.  Then, the area of the geodesic ellipsoidal shell $\tilde{\mathfrak{E}}_{\epsilon}$ is given by
\begin{subequations}
\begin{align}
	\Area(\tilde{\mathfrak{E}}_{\epsilon}) & = \int_{\tilde{\mathfrak{E}}_{\epsilon}}\rmd\mu = \int_{\mathfrak{E}_{\epsilon}}\exp_{\nu}^{*}(\rmd\mu)\\
	& = \int_{\Orb(\perm)}\int_{\Ellipse(\epsilon)}(1+O(\epsilon)^{2})\omega(X_{1},\dots,X_{N^{2}})\rmd V\wedge\rmd P\\
	& = (1+O(\epsilon^{2}))\epsilon^{N^{2}-\smd-1}\int_{\Orb(\perm)}\int_{\Ellipse(1)}\omega(X_{1},\dots,X_{N^{2}})\rmd V\wedge\rmd P\\
	& = \epsilon^{N^{2}-\smd-1}\Vol_{\rmHS}\big(\Orb(\perm)\big)\Vol(\Ellipse(1)) + O(\epsilon^{N^{2}-\smd+1})
\end{align}
\end{subequations}
where $\Ellipse(r)$ is the ellipse with principal axes $r/|\beta_{i}|$.

So we conclude that for small enough $\epsilon>0$, the volume of the ellipsoidal tube about $\Orb(\perm)$ is 
\begin{subequations}
\begin{align}
	\Vol(\mathfrak{T}_{\epsilon}) & = \int_{0}^{\epsilon}\Area(\tilde{\mathfrak{E}}_{r})\rmd r\\
	& = \frac{\epsilon^{N^{2}-\smd}}{N^{2}-\smd}\Vol_{\rmHS}\big(\Orb(\perm)\big)\Vol(\Ellipse(1)) + O(\epsilon^{N^{2}-\smd+2})\\
%	& = \int_{\Orb(\perm)}\frac{(\pi \epsilon^{2})^{\frac{N^{2}-\smd}{2}}}{\Gamma\big((N^{2} - \smd)/2 + 1\big)}\Big(\Vol_{\rmHS}\big(\Orb(\perm)\big) + O\big(\epsilon^{2}\big)\Big)\\
	& = \frac{2^{\frac{d+N}{2}}\pi^{\frac{N(N+1)}{2}}\epsilon^{N^{2}-\smd}}{\Gamma\big(\frac{N^{2} - \smd}{2} + 1\big)\prod|\beta_{i}|}\times\frac{\prod_{ij}\prod_{r=0}^{k_{ij}-1}r!}{\prod_{i}\prod_{p=0}^{n_{i}-1}p!\prod_{j}\prod_{q=0}^{m_{j}-1}q!} + O\big(\epsilon^{N^{2} - \smd + 2}\big),
\end{align}
\end{subequations}
where $\smd = \sum n_{i}^{2} + \sum m_{j}^{2} \sum k_{ij}^{2} = \dim\big(\Orb(\perm)\big)$.  Then the volume fraction of the tube within $\UN$ is 
\begin{subequations}
\begin{align}
	\VolFrac(\mathfrak{T}_{\epsilon}) & = \frac{\Vol(\mathfrak{T}_{\epsilon})}{\Vol_{\rmHS}\big(\UN\big)} = \frac{\prod_{s=0}^{N-1}s!}{(2\pi)^{\frac{N(N+1)}{2}}}\Vol(\mathfrak{T}_{\epsilon})\\
	& = \frac{ \epsilon^{N^{2} - \smd}}{2^{\frac{N^{2} - \smd}{2}}\Gamma\big(\frac{N^{2} - \smd}{2} + 1\big)\prod|\beta_{i}|}\times\frac{\prod_{s=0}^{N-1}s!\prod_{ij}\prod_{r=0}^{k_{ij}-1}r!}{\prod_{i}\prod_{p=0}^{n_{i}-1}p!\prod_{j}\prod_{q=0}^{m_{j}-1}q!} + O\big(\epsilon^{N^{2} - \smd + 2}\big).\label{eqn:estVolFrac}
\end{align}
\end{subequations}

\bigskip

\subsection{Examples}
\label{sec:volFracExamples}
We now return to the examples considered in Section \ref{sec:volExamples} and compute the volume fractions of the corresponding near-critical sets.
\begin{example}[Maximum Submanifold of $P_{i\to f}$]
	The Hessian of $P_{i\to f}$ on the maximum submanifold has rank $2N-2$, and all of the nonzero eigenvalues are $\beta_{i}=-1$.  Then, the volume of $\Ellipse(1)$ with principal axes $1/|\beta_{i}|$ is just the volume of the unit sphere $S^{2N-2}$, which is $(2N-1)\frac{2^{N}\pi^{N-1}}{(2N-1)!!} = \frac{2^{2N-1}\pi^{N-1}(N-1)!}{(2N-2)!}$ (the double factorial $(2N-1)!!$ is defined as the product of the odd integers from 1 to $2N-1$).  Using the volume computed in Section \ref{sec:volExamples} for $\Orb(\mathbb{I})$, we find
	\begin{subequations}
	\begin{align}
		\VolFrac(\mathfrak{T}_{\epsilon}) & = \frac{\epsilon^{2N-2}}{2N-2}\frac{\Vol_{\rmHS}\Orb(\mathbb{I})\Vol(\Ellipse(1))}{\Vol_{\rmHS}\big(\UN\big)} + O(\epsilon^{2N})\\
		& = \frac{\epsilon^{2N-2}}{2N-2}\frac{(2\pi)^{\frac{1}{2}(N^{2}-N+2)}}{\prod_{p=0}^{N-2}p!}\frac{2^{2N-1}\pi^{N-1}(N-1)!}{(2N-2)!}\frac{\prod_{s=0}^{N-1}s!}{(2\pi)^{\frac{N(N+1)}{2}}} + O(\epsilon^{2N})\\
		& = \epsilon^{2N-2}\frac{2^{N-1}(N-1)!(N-2)!}{(2N-2)!} + O(\epsilon^{2N})\\
		& = \epsilon^{2N-2}\frac{(N-2)!}{(2N-3)!!} + O(\epsilon^{2N}),
	\end{align}
	\end{subequations}
where \begin{equation}\left(\frac{1}{2}\right)^{N-3}\frac{1}{2N-3} < \frac{(N-2)!}{(2N-3)!!} \leq \left(\frac{2}{3}\right)^{N-3}\frac{1}{2N-3}\end{equation} for $N>2$.
\end{example}
\begin{example}[Minimum Submanifold of $P_{i\to f}$]
	The Hessian of $P_{i\to f}$ on the minimum submanifold has rank $2$, with both nonzero eigenvalues equal to one.  So the volume of $\Ellipse(1)$ is just the volume of the unit sphere $S^{2}$, which is $4\pi$.  Using the volume computed in Section \ref{sec:volExamples} for $\Orb(\perm)$, we find		
	\begin{subequations}
	\begin{align}
		\VolFrac(\mathfrak{T}_{\epsilon}) & = \frac{\epsilon^{2}}{2}\frac{\Vol_{\rmHS}\Orb(\perm)\Vol(\Ellipse(1))}{\Vol_{\rmHS}\big(\UN\big)} + O(\epsilon^{4})\\
		& = \frac{\epsilon^{2}}{2}\frac{(2\pi)^{\frac{1}{2}(N^{2}+N-2)}}{(N-2)!\prod_{p=0}^{N-2}p!}(4\pi)\frac{\prod_{s=0}^{N-1}s!}{(2\pi)^{\frac{N(N+1)}{2}}} + O(\epsilon^{4})\\
		& = (N-1)\epsilon^{2} + O(\epsilon^{4}).
	\end{align}
	\end{subequations}
\end{example}

\begin{example}[Fully Non-degenerate $\rho$ and $\Ob$]
	In the case where $\rho$ and $\Ob$ are both fully non-degenerate, $m_{1} = \dots = m_{N} = 1$, $n_{1} = \dots = n_{N} = 1$ and $\mathbf{K} = \perm^{\dag}$.  Then for any critical submanifold $\Orb(\perm)$, all of which are $N$ dimensional tori,
	\begin{subequations}
	\begin{align}
		\VolFrac(\mathfrak{T}_{\epsilon}) & = \frac{\epsilon^{N^{2}-N}}{N^{2}-N}\frac{\Vol_{\rmHS}\Orb(\perm)\Vol(\Ellipse(1))}{\Vol_{\rmHS}\big(\UN\big)} + O(\epsilon^{N^{2}-N+2})\\
		& = \frac{\epsilon^{N^{2}-N}\pi^{\frac{N^{2}-N}{2}}(2\pi)^{N}}{\left(\frac{N^{2}-N}{2}\right)!\prod |\beta_{i}|}\frac{\prod_{s=0}^{N-1}s!}{(2\pi)^{\frac{N(N+1)}{2}}} + O(\epsilon^{N^{2}-N+2})\\
		& = \frac{\prod_{s=0}^{N-1}s!}{2^{\frac{N(N-1)}{2}}\left(\frac{N^{2}-N}{2}\right)!\prod |\beta_{i}|}\epsilon^{N^{2}-N} + O(\epsilon^{N^{2}-N+2})
%		& = \frac{(N^{2}-N-1)!!\prod_{s=0}^{N-1}s!}{(N^{2}-N)!\prod |\beta_{i}|}\epsilon^{N^{2}-N} + O(\epsilon^{N^{2}-N+2})
	\end{align}
	\end{subequations}
where $\prod|\beta_{i}|$ is the product of the $N^{2}-N$ nonzero Hessian eigenvalues, and where \begin{equation}\frac{1}{\prod_{s=1}^{N-1}(s^{2}-s+2)^{s}} \leq \frac{\prod_{s=0}^{N-1}s!}{2^{\frac{N(N-1)}{2}}\left(\frac{N^{2}-N}{2}\right)!} \leq \frac{1}{\prod_{s=1}^{N-1}(s+1)^{s}} = \prod_{s=1}^{N-1}\frac{s!}{N!}.\end{equation}
\end{example}

\bigskip

\section{Asymptotic Analysis}
\label{sec:asymptoticAnalysis}
The expression in \eqref{eqn:estVolFrac} provides a means of estimating the volume fraction of any given near-critical set.  However, for probing the general asympotic behavior of these volumes as the dimension $N$ of the state space rises, this estimate is inadequate since it only holds for small enough $\epsilon >0$, where ``small enough'' is determined for each system and each dimension $N$.  Before seeking a new expression of practical utility, it is necessary to define the parameters of the desired asymptotic analysis.  Fix some $N_{0}>0$, $N_{0}\times N_{0}$ density matrix $\rho^{0}$, and $N_{0}\times N_{0}$ Hermitian observable operator $\Ob^{0}$.  Then for any $z\in\mathbb{N}$, let $N_{z} = N_{0}+z$, $\rho^{z} = \rho^{0}\oplus 0_{z}$, and $\Ob^{z} = \Ob^{0}\oplus 0_{z}$, where $0_{z}$ is the $z\times z$ zero matrix.  Then each critical value of the kinematic landscape $J_{z}(U) = \Tr(U\rho^{z}U^{\dag}\Ob^{z})$ is also a critical value of $J_{z+1}$, so each critical submanifold of $J_{z}$ has a direct analog in $J_{z+1}$.  In this fashion, one can decribe an infinite sequence of critical submanifolds as $z\to\infty$ and consider the asymptotic behavior of the volume fractions of the near-critical sets around these critical submanifolds.  It will be argued that these volume fractions converge to zero as $z\to\infty$.  The landscape $J_{z}$ for $z>N_{0}$ has the same number of critical submanifolds as $J_{N_{0}}$, so if the volume fractions of the individual near-critical sets converge to zero, then the total near-critical volume fraction of $J_{z}$ also converges to zero as $z\to \infty$.  

Because compact Lie groups with bi-invariant metrics have non-negative sectional curvature \cite{doCarmo1992}, the following comparison theorem proved in \cite[Ch. 8]{Gray2004} may be used to bound the volume of spherical tubes about a submanifold $P$.
\begin{theorem}Let $M$ be an $n$-dimensional Riemannian manifold with non-negative sectional curvature.  Then for any $\smd$-dimensional submanifold $P\subset M$ and all $r\geq 0$, the volume of the spherical tube of radius $r$ about $P$ in $M$ is bounded as \begin{equation}V_{P}^{M}(r)\leq \int_{0}^{r}\int_{P}\int_{S^{n-\smd-1}}t^{n-\smd-1}\max\left(\left(1-\frac{t}{\smd}\langle H, u\rangle\right)^{\smd},0\right)\rmd u\; \rmd P\; \rmd t,\end{equation}where $H$ is the mean curvature vector field.\end{theorem}
Since the submanifolds of $\UN$ considered in this paper are all minimal ($H = 0$), this bound reduces to \begin{equation}V_{\Orb(\perm)}^{\UN}(r)\leq \frac{r^{N^{2}-\smd}}{N^{2}-\smd}\Vol\big(S^{N^{2}-\smd-1}\big)\Vol\big(\Orb(\perm)\big).\end{equation} Of course, any ellipsoidal tube with longest principal axis $\epsilon/|\beta_{\min}|$ is contained in the spherical tube with radius $r = \epsilon/|\beta_{\min}|$.  In the sequence of density matrices and observable operators $\rho^{z}$ and $\Ob^{z}$ described above, once $z>N_{0}$, the set of distinct Hessian eigenvectors for a given critical submanifold does not change with $z$.  Then $\beta_{\min}$ is fixed and the $\epsilon$-ellipsoidal tube about the critical submanifold is contained in the spherical tube with radius $r = \epsilon/|\beta_{\min}|$, so the volume of the $\epsilon$-ellipsoidal tube is bounded by 
\begin{equation}\Vol(\mathfrak{T}_{\epsilon})\leq \frac{\epsilon^{N^{2}-\smd}}{(N^{2}-\smd)|\beta_{\min}|^{N^{2}-\smd}}\Vol\big(S^{N^{2}-\smd-1}\big)\Vol\big(\Orb(\perm)\big),\end{equation} 
so that 
\begin{subequations}
\begin{align}
	\VolFrac(\mathfrak{T}_{\epsilon}) & \leq \frac{\epsilon^{N^{2}-\smd}}{(N^{2}-\smd)|\beta_{\min}|^{N^{2}-\smd}}\frac{\Vol\big(S^{N^{2}-\smd-1}\big)\Vol\big(\Orb(\perm)\big)}{\Vol(\UN)}\\
%	&  = \frac{\epsilon^{N^{2}-\smd}}{(N^{2}-\smd)|\beta_{\min}|^{N^{2}-\smd}}\frac{2\pi^{(N^{2}-\smd)/2}}{\left(\frac{N^{2}-\smd}{2}-1\right)!}(2\pi)^{\frac{1}{2}\big(\smd + N\big)}\frac{\prod_{ij}\prod_{r=0}^{k_{ij}-1}r!}{\prod_{i}\prod_{p=0}^{n_{i}-1}p!\prod_{j}\prod_{q=0}^{m_{j}-1}q!}\frac{\prod_{s=0}^{N-1}s!}{(2\pi)^{\frac{N(N+1)}{2}}}\\
	&  = \frac{\epsilon^{N^{2}-\smd}}{2^{(N^{2}-\smd)/2}\left(\frac{N^{2}-\smd}{2}\right)!|\beta_{\min}|^{N^{2}-\smd}}\frac{\prod_{s=0}^{N-1}s!\prod_{ij}\prod_{r=0}^{k_{ij}-1}r!}{\prod_{i}\prod_{p=0}^{n_{i}-1}p!\prod_{j}\prod_{q=0}^{m_{j}-1}q!}\label{eqn:ellipsoidalTubeBound}
\end{align}
\end{subequations}
where $\smd = \sum m_{i}^{2} + \sum n_{j}^{2} - \sum k_{ij}^{2}$ is the dimension of the critical submanifold.

Now, for $z>N_{0}$, the set of critical values does not change with $z$, so fix a critical value $v$ and consider the contingency table associated with the corresponding critical submanifold.  As $z$ increases, the only elements in the table that will change are $m_{s}$, the multiplicity of the zero eigenvalue of $\rho$, $n_{r}$, the multiplicity of the zero eigenvalue of $\Ob$, and $k_{sr}$, the degree of overlap between the zero eigenvalues of $\Lambda$ and $\perm^{\dag}\Sigma\perm$, where $\Lambda = \Omega\rho\Omega^{\dag}$ and $\Sigma = \Gamma\mathcal{O}\Gamma^{\dag}$ are diagonalizations of $\rho$ and $\Ob$ with decreasing elements.  For $z>N_{0}$, these indices change as $m_{s}(z+1) = m_{s}(z) + 1$, $n_{r}(z+1) = n_{r}(z)+1$, and $k_{sr}(z+1) = k_{sr}(z)+1$.  Then $\smd(z+1) = \smd(z) + 2(m_{s}(z) + n_{r}(z) - k_{sr}(z)) + 1 = \smd(z) + 2(m_{s}(N_{0}) + n_{r}(N_{0}) - k_{sr}(N_{0})) + 2(z-N_{0}) + 1$.  So $N_{z+1}^{2}-\smd(z+1) = N_{z}^{2}-\smd(z) + 2(2N_{0} - m_{s}(N_{0}) - n_{r}(N_{0}) + k_{sr}(N_{0}))$.  Let $D^{z}(\epsilon)$ denote the right-hand side of \eqref{eqn:ellipsoidalTubeBound} for $z>N_{0}$.  Then let 
\begin{subequations}
\begin{align}
	F^{z}(\epsilon) & = \frac{D^{z+1}(\epsilon)}{D^{z}(\epsilon)} = \frac{\epsilon^{2\zeta}}{2^{\zeta}|\beta_{\min}|^{2\zeta}}\frac{\left(\frac{N_{z}^{2}-\smd(z)}{2}\right)!}{\left(\frac{N_{z}^{2}-\smd(z)}{2} + \zeta\right)!}\frac{N_{z}!k_{sr}(z)!}{n_{r}(z)!m_{s}(z)!}\\
	& = \frac{\epsilon^{2\zeta}}{2^{\zeta}|\beta_{\min}|^{2\zeta}}\frac{\left(\frac{4N_{0}^{2}-\smd(N_{0})}{2} + (z-N_{0})\zeta\right)!}{\left(\frac{4N_{0}^{2}-\smd(N_{0})}{2} + (z+1-N_{0})\zeta\right)!}\frac{N_{z}!k_{sr}(z)!}{n_{r}(z)!m_{s}(z)!}
\end{align}
\end{subequations}
where $\zeta = 2N_{0} - m_{s}(N_{0}) - n_{r}(N_{0}) + k_{sr}(N_{0})\geq 0$ since if $\nu_{\rho}$ and $\nu_{\Ob}$ are the nullities of $\rho^{0}$ and $\Ob^{0}$, then $m_{s}(N_{0}) = N_{0} + \nu_{\rho}$, $n_{r}(N_{0}) = N_{0} + \nu_{\Ob}$, and $k_{sr}(N_{0})\geq \nu_{\rho} + \nu_{\Ob}$.  Moreover, note that if $\zeta = 0$, then $k_{sr}(N_{0}) = \nu_{\rho} + \nu_{\Ob}$ and the critical value under consideration must be $v=0$, so for any critical value $v\neq 0$, $\zeta >0$.  Let \begin{equation}G^{z} = \frac{F^{z}(\epsilon)}{F^{z-1}(\epsilon)} = \frac{\left(\frac{4N_{0}^{2}-\smd(N_{0})}{2} + (z-1-N_{0})\zeta\right)!}{\left(\frac{4N_{0}^{2}-\smd(N_{0})}{2} + (z+1-N_{0})\zeta\right)!}\frac{(N_{0}+z)(k_{sr}(N_{0})-N_{0}+z)}{(n_{r}(N_{0})-N_{0}+z)(m_{s}(N_{0})-N_{0}+z)},\end{equation} where  \begin{equation}\frac{(N_{0}+z)(k_{sr}(N_{0})-N_{0}+z)}{n_{r}(N_{0})-N_{0}+z)(m_{s}(N_{0})-N_{0}+z)} = 1 + \zeta\frac{1}{z} + O(z^{-2}).\label{eqn:zetaFracSeries}\end{equation} Then whenever $\zeta>0$, the expression in \eqref{eqn:zetaFracSeries} converges to 1 for large $z$, so that $G^{z}\searrow 0$ as $z\to\infty$.  Therefore $F^{z}(\epsilon)\searrow 0$ and consequently $D^{z}(\epsilon)\searrow 0$ as $z\to\infty$.  Moreover, it may be seen that $G^{z} = O(z^{-2\zeta})$, so that $F^{z}(\epsilon) = O((z!)^{-2\zeta})$ and $D^{z}(\epsilon) = O(\prod_{s = 1}^{z-1}(s!)^{-2\zeta})$.  So the volume fraction of the spherical tube of radius $\epsilon/|\beta_{\min}|$ converges to zero as  $z\to \infty$ and finally we may conclude that the volume fraction of the ellipsoidal tube approximating the $\epsilon$ near-critical set about a critical submanifold with critical value $v\neq 0$ also converges to zero as $z\to\infty$.  Since this convergence proceeds very quickly as the negative power of a product of factorials, it is independent of the slower exponential contribution from $\epsilon^{N^{2}-\smd}$, and therefore independent of the choice of $\epsilon$.

This last result demonstrates the convergence of the volume fractions of the approximating ellipsoidal tubes.  To improve on this and show the convergence of the volume fractions of the near-critical sets themselves, it will be necessary to make use of the following conjecture.  Evidence in support of the conjecture is presented in Appendix \ref{sec:supportConj}.

\begin{conjecture}
\label{conj:gradientBound}
Let $U\in\UN$ be a critical point of $J(U) = \Tr(U\rho U^{\dag}\Ob)$, and let $A\in\un$ be such that the tangent vector $UA\in\rmT_{U}\UN$ is of unit length and lies normal to the critical submanifold through $U$.  Define $f:[0,\pi/(2\sqrt{2})]\to\mathbb{R}$ to be the norm squared of the gradient of $J$ along the unit speed geodesic in the direction $UA$, i.e. \begin{equation}f(s) = \big\|\grad J(U\exp(sA))\big\|^{2} = \big\|[\exp(-sA)U^{\dag}\Ob U\exp(sA), \rho]\big\|^{2} = \big\|[U^{\dag}\Ob U, \exp(sA)\rho\exp(-sA)]\big\|^{2},\end{equation} and let $\beta_{\min}$ be the minimum (in absolute value) nonzero eigenvalue of the Hessian of $J$ at $U$.  Then $f(s)\geq \beta_{\min}^{2}\sin^{2}(\sqrt{2}s)/2$ for all $0\leq s\leq \pi/(2\sqrt{2})$.
\end{conjecture}

With this conjecture, once $z>N_{0}$, the set of Hessian eigenvalues no longer changes with $z$, so there is one fixed $\beta_{\min}$ for all $z>N_{0}$.  Then for any $\epsilon<|\beta_{\min}|/\sqrt{2}$, when $r = \epsilon\pi/(2|\beta_{\min}|)<\pi/(2\sqrt{2})$, it is found that $f(r) \geq \beta_{\min}^{2}\sin^{2}(\sqrt{2}r)/2 \geq 4\beta_{\min}^{2}r^{2}/\pi^{2} = \epsilon^{2}$, so that the $\epsilon$ near-critical set about the critical submanifold is contained with the radius $r$ spherical tube.  For any fixed radius, the volume fraction of this spherical tube was shown to converge to zero as $z\to \infty$ when the critical value $v\neq 0$.  So, it is seen that if Conjecture \ref{conj:gradientBound} holds, then the volume fraction converges to zero (as the negative power of a product of factorials) for the $\epsilon$ near-critical set about any critical submanifold with critical value $v\neq 0$.  Since saddles are non-attractive critical points, this suggests that as $z$ gets large, the probability becomes vanishingly small that the gradient flow from a (Haar distributed) random point passes through one of these flat ``near-critical'' regions around a saddle submanifold.

Referring back to the examples from Sections \ref{sec:volExamples} and \ref{sec:volFracExamples}, it may be seen that, since the transition probability $P_{i\to f}$ involves a rank one density matrix $\rho$ and observable operator $\Ob$, there exist sequences of these landscapes (for increasing $N$) that fit the required behavior for $\rho^{z}$ and $\Ob^{z}$ needed for the analysis in this section.  Of the two critical submanifolds of this landscape, only the maximum submanifold $P_{i\to f} = 1$ satisfies the further condition that $\zeta$ be nonzero.  The remaining example of non-degenerate $\rho$ and $\Ob$ for every $N$ does not adhere to these requirements and falls outside this analysis.  Ultimately, the asymptotic analysis of such an example in which $\rho^{N}$ and $\Ob^{N}$ are fully non-degenerate for all $N$ would be sensitive to the asymptotic behavior of the eigenvalues of $\rho^{N}$ and $\Ob^{N}$.

\bigskip

\section{Conclusions}
\label{sec:conclusions}
This work computed the volumes of the critical submanifolds of $J(U) = \Tr(U\rho U^{\dag}\Ob)$ in the induced Hilbert-Schmidt measure, and developed estimates and bounds  for the volume fractions of near-critical sets of the form $C_{\epsilon}^{\perm} = \{U \;:\; \|\grad J(U)\|\leq \epsilon\}\subset\UN$.  An asymptotic analysis of these volume fractions revealed that, when the critical value is non-zero, the volume fraction converges to zero as $N\to\infty$.  This result helps to explain previous observations that numerical quantum optimal control experiments seem not to be adversely affected by the presence of a large number of high-dimensional saddle submanifolds.

The work presented here focussed on the geometry of the kinematic landscape $J:\UN\to\mathbb{R}$.  Although it is outside the scope of this paper, to relate this work more closely to numerical and laboratory quantum optimal control experiments, these results should be pulled back to the corresponding dynamical landscape $\tilde{J}:\CtlSp\to\mathbb{R}$ defined on the space of controls.  This effort must address a number of difficulties including the dependence of these landscapes on the details of the quantum system and, depending on definitions, dependence on the final time $T$.  But perhaps the biggest problem is that, in order to make mathematical sense of the concept of volume fractions, a probability measure needs to be defined either explicitly or implicitly on the control space $\CtlSp$, which is typically infinite dimensional and unbounded, such as $\CtlSp = L^{2}([0,T];\mathbb{R})$.  Overcoming these difficulties could provide a clearer picture of the gradient flow of $\tilde{J}$ on $\CtlSp$.

\bigskip

\section{Acknowledgments}
This work was supported, in part, by U.S. Department of Energy (DOE) Contract No. DE-AC02-76-CHO-3073 through the Program in Plasma Science and Technology at Princeton.    We also acknowledge support from Lockheed Martin and from DOE grant No. DE-FG02-02ER15344

\bigskip

\appendix
\section{The Chevalley Lattice of $\mathrm{U}(a_{1})\oplus\dots\oplus\mathrm{U}(a_{b})$\label{sec:ChevalleyLattice}}
The first step to evaluating the volume of a compact Lie group $G$ via Macdonald's formula \cite{Macdonald1980, Hashimoto1997} is to work out it's Chevalley lattice $\gZ$ (closely related to the concept of the Chevalley basis \cite{Humphreys1972}).  In this appendix, we describe this lattice for groups of the form $G = \mathrm{U}(a_{1})\oplus\dots\oplus\mathrm{U}(a_{b})$ as we will need to compute volumes of two such Lie groups to obtain the desired structure of the critical set $\Orb(\perm)$.  To that end, first observe that the set of all diagonal matrices in $G$ forms a maximal abelian torus $T$, the Lie algebra of which, $\mathfrak{t}$, is the set of all diagonal matrices in $\mathfrak{g} = \mathrm{u}(a_{1})\oplus\dots\oplus\mathrm{u}(a_{b})$.  Then defining $\tZ$ such that $2\pi\tZ$ is the kernel of $\exp:\mathfrak{t}\to T$, it is found that $\tZ = i\diag(\mathbb{Z}^{\bar{a}})$ is the lattice of diagonal $\bar{a}\times \bar{a}$ matrices with imaginary integer elements, where $\bar{a} = \sum a_{j}$.  

%	Now, each $\mathrm{u}(a_{l})$ can be written as $\mathrm{su}(a_{l})\oplus A$, where $A$ is a one-dimensional abelian algebra, and the complexification o
	Now, each $\mathrm{u}(a_{l})$ has complexification $\mathrm{gl}(a_{l})$ which is the direct sum of $\mathrm{sl}(a_{l})$ and a one-dimensional abelian algebra.  So the root vectors of $\mathrm{gl}(a_{l})$ are just those of $\mathrm{sl}(a_{l})$.  The root space decomposition of $\mathrm{sl}(a_{l})$ is defined by the positive roots $\alpha_{jk}^{l}$ ($1\leq j<k\leq a_{l}$) and the corresponding coroots $H_{\alpha_{jk}^{l}}$ and root vectors $X_{\alpha_{jk}^{l}}$ \cite{Hall2003} given by 
	\begin{subequations}
	\begin{align}
		\alpha_{jk}^{l}(H) & = H_{jj} - H_{kk} & 1\leq j<k\leq a_{l}\\
		H_{\alpha_{jk}^{l}} & = |j\rangle\langle j| - |k\rangle\langle k| \\
		X_{\alpha_{jk}^{l}} & = |j\rangle\langle k|\\
		X_{-\alpha_{jk}^{l}} & = -|k\rangle\langle j|.
	\end{align}
	\end{subequations}
where $H$ is an arbitrary element of the Cartan subalgebra of diagonal matrices in $\mathrm{sl}(a_{l})$.  If, for any $X = \xi + i\eta \in\mathrm{sl}(a_{l})$ with $\xi,\eta\in\mathrm{su}(a_{l})$, we define $\bar{X} = \xi - i\eta$, then the $X_{\alpha}$'s above satisfy $X_{-\alpha} = \bar{X}_{\alpha}$ as indicated in \cite{Macdonald1980}.  Then finally let 
	\begin{subequations}
	\begin{align}
		\xi_{\alpha_{jk}^{l}} & := \big(\underbrace{0,\dots,0}_{l-1\text{ zeros}},X_{\alpha_{jk}^{l}} + X_{-\alpha_{jk}^{l}},\underbrace{0,\dots,0}_{b-l\text{ zeros}}\big) = \big(0,\dots,0,|j\rangle\langle k| - |k\rangle \langle j|,0,\dots,0\big)\\
		\eta_{\alpha_{jk}^{l}} & := \big(0,\dots,0,i(X_{\alpha_{jk}^{l}} - X_{-\alpha_{jk}^{l}}),0,\dots,0\big) = \big(0,\dots,0,i\big(|j\rangle\langle k| + |k\rangle \langle j|\big),0,\dots,0\big)
	\end{align}
	\end{subequations}
for all $1\leq l \leq b$ and all $1\leq j < k \leq a_{l}$.  These vectors, along with the basis of $\tZ$ given by $\{\tau_{j}:=i |j\rangle\langle j|\}$ for $j=1,\dots,\bar{a}$, form the basis for the Chevalley lattice of $\mathrm{U}(a_{1})\oplus\dots\oplus\mathrm{U}(a_{b})$ denoted by $\gZ$ in \cite{Macdonald1980}.  The volume of the fundamental cell $\mathfrak{g}/\gZ$ with respect to a given inner product is just the square root of the determinant of the Gram matrix constructed from the basis for $\gZ$.  

\bigskip

\section{The Hessian of $J(U) = \Tr(U\rho U^{\dag}\mathcal{O})$\label{sec:HessianStructure}}
In this appendix, the Hessian operator of $J$ is described and its eigenvalues and eigenvectors obtained.  First, note that, since $U^{\dag}U = \mathbb{I}$, $\delta U^{\dag} U + U^{\dag}\delta U = 0$ for any $\delta U \in T_{U}\UN$, so $\delta U^{\dag} = - U^{\dag}\delta U U^{\dag}$.  Let $J:\UN\rightarrow \mathbb{R}$ be given by $J(U) = \Tr(U\rho U^{\dag}\mathcal{O})$.  Then 
	\begin{subequations}
	\begin{align}
		\rmd_{U}J(\delta U) & = \Re\Tr(\delta U \rho U^{\dag} \mathcal{O} + U\rho \delta U^{\dag}\mathcal{O}) = \Re\Tr(\delta U\rho U^{\dag}\mathcal{O} - U\rho U^{\dag}\delta U U^{\dag}\mathcal{O})\\
		& = \Re\Tr([\rho, U^{\dag}\mathcal{O}U]U^{\dag}\delta U) = \langle U[U^{\dag}\mathcal{O}U,\rho], \delta U\rangle
	\end{align}
	\end{subequations}
so that $\grad J(U) = U[U^{\dag}\mathcal{O}U,\rho]\in T_{U}\UN$ and $\grad J$ is a vector field over $\UN$, i.e. $\grad J\in\mathfrak{X}\big(\UN\big)$ [the $C^{\infty}$ module of smooth vector fields on $\UN$].  For any $X\in\mathfrak{X}\big(\UN\big)$, the Hessian of $J$ is given by $\Hess(X) = \nabla_{X}\grad J$.  Because $\UN$ is endowed with the Riemannian metric induced from the real Hilbert-Schmidt inner product on $\mathbb{C}^{N\times N}$, $\nabla_{X}\grad J = \big(\rmd(\grad J)(\overline{X})\big)^{\mathscr{T}}$ \cite{doCarmo1992}, where $\overline{X}$ is an extension of $X$ to $\mathbb{C}^{N\times N}$ and $\mathscr{T}$ denotes the tangential part.  Now, \begin{equation}\rmd_{U}\grad J(\delta U) = \delta U[U^{\dag}\mathcal{O}U,\rho] + U[\delta U^{\dag}\mathcal{O}U,\rho] + U[U^{\dag}\mathcal{O}\delta U,\rho]\end{equation}
and therefore, 
	\begin{subequations}
	\begin{align}
		\Hess(X) & = \nabla_{X}\grad J = \frac{1}{2}U\Big(U^{\dag}X[U^{\dag}\mathcal{O}U,\rho] + [U^{\dag}\mathcal{O}U,\rho]X^{\dag} U + 2[X^{\dag}\mathcal{O}U,\rho] + 2[U^{\dag}\mathcal{O}X,\rho]\Big)\\
		& = \frac{1}{2}U\Big(U^{\dag}X[U^{\dag}\mathcal{O}U,\rho] - [U^{\dag}\mathcal{O}U,\rho]U^{\dag}X + 2[X^{\dag}UU^{\dag}\mathcal{O}U,\rho] + 2[U^{\dag}\mathcal{O}UU^{\dag}X,\rho]\Big)\\
		& = \frac{1}{2}U\Big(U^{\dag}X[U^{\dag}\mathcal{O}U,\rho] - [U^{\dag}\mathcal{O}U,\rho]U^{\dag}X - 2[U^{\dag}X U^{\dag}\mathcal{O}U,\rho] + 2[U^{\dag}\mathcal{O}UU^{\dag}X,\rho]\Big)\\
		& = U\left(\frac{1}{2}[U^{\dag}X,[U^{\dag}\mathcal{O}U,\rho]] + [[U^{\dag}\mathcal{O}U,U^{\dag}X],\rho]\right).
	\end{align}
	\end{subequations}
Any Hessian may be written as $\langle Y, \Hess(X)\rangle = XYJ - \big(\nabla_{X}Y\big)J$.  Note that $XYJ$ is second order in $J$ and first order in the manifold (i.e., relates to tangent spaces), while $\big(\nabla_{X}Y\big)J$ is first order in $J$ and second order in the manifold (i.e., relates to curvature of the manifold), so this second term vanishes at critical points and on flat spaces (when $X$ and $Y$ are constant vector fields).  In the above case, for left invariant $X$ and $Y$, it can be shown that $XYJ = \left\langle Y, U[[U^{\dag}\mathcal{O}U,U^{\dag}X],\rho]\right \rangle$ and $\big(\nabla_{X}Y\big)J = -\left\langle Y,\frac{1}{2}U[U^{\dag}X,[U^{\dag}\mathcal{O}U,\rho]]\right\rangle$.

Suppose that $U$ is a critical point of $J$, so that $\grad J(U) = U[U^{\dag}\mathcal{O}U,\rho] = 0$ and $\Hess(X) = U[[U^{\dag}\mathcal{O}U,U^{\dag}X],\rho]$.  Let $\Omega$ diagonalize $\rho$ and $\Gamma$ diagonalize $\mathcal{O}$ such that the diagonal elements of $\Lambda = \Omega\rho\Omega^{\dag}$ and $\Sigma = \Gamma\mathcal{O}\Gamma^{\dag}$ are decreasing.  Then $U = \Gamma^{\dag} V\perm W^{\dag}\Omega$ for some $(V,W)\in\Umn$ and some $\perm\in S_{N}$ \cite{Wu2008}.  We can then write
	\begin{equation}\Hess(UY) = U[[\Omega^{\dag}W\perm^{\dag}\Sigma\perm W\Omega, Y], \Omega^{\dag}\Lambda\Omega] =  U\Omega^{\dag}W[[\perm^{\dag}\Sigma\perm, W^{\dag}\Omega Y\Omega^{\dag}W], \Lambda]\Omega W^{\dag}.\end{equation}
Then $\Hess(UY) = \beta UY$ if and only if
	\begin{subequations} 
	\begin{align}
		\beta\tilde{Y} & = [[\perm^{\dag}\Sigma\perm, \tilde{Y}], \Lambda]\\
		\beta\tilde{Y}_{jk} & = -(\lambda_{j} - \lambda_{k})(\sigma_{\perm(j)} - \sigma_{\perm(k)})\tilde{Y}_{jk}
	\end{align}
	\end{subequations}
where $\tilde{Y} := W^{\dag}\Omega Y\Omega^{\dag}W$.  So either $\beta = -(\lambda_{j}- \lambda_{k})(\sigma_{\perm(j)} - \sigma_{\perm(k)})$ or $\tilde{Y}_{jk} = 0$.  Then an orthonormal set of solutions is given by
	\begin{subequations}
	\label{eqn:eigenvectors}
	\begin{align}
		\beta & = 0 & \tilde{Y} = i|l\rangle \langle l|\\
		\beta & = -(\lambda_{j} - \lambda_{k})(\sigma_{\perm(j)} - \sigma_{\perm(k)}) & \tilde{Y} = \frac{i}{\sqrt{2}}\big(|j\rangle\langle k| + |k\rangle\langle j|\big)\\
		\beta & = -(\lambda_{j} - \lambda_{k})(\sigma_{\perm(j)} - \sigma_{\perm(k)}) & \tilde{Y} = \frac{1}{\sqrt{2}}\big(|j\rangle\langle k| - |k\rangle\langle j|\big)
	\end{align}
	\end{subequations}
for $1\leq l\leq N$ and $1\leq j<k\leq N$.  These $\beta$ and the corresponding vectors $UY = U\Omega^{\dag}W\tilde{Y}W^{\dag}\Omega$ are the eigenvalues and eigenvectors of $\Hess$ at the critical point $U$.

\bigskip

\section{The Shape Operators of the Critical Submanifolds}
\label{sec:shapeOperators}
Consider a critical submanifold for some permutation matrix $\perm$ and degeneracy structures $\mathbf{m}$ and $\mathbf{n}$ of $\Ob$ and $\rho$, respectively.  Then, let $\Omega\in\UN$ and $\Gamma\in\UN$ diagonalize $\rho$ and $\Ob$, so that $\Omega\rho\Omega^{\dag} = \Lambda$ and $\Gamma\Ob\Gamma^{\dag}$ with $\lambda_{1}\geq\dots\geq \lambda_{N}$ and $\sigma_{1}\geq\dots\geq\sigma_{N}$.  The critical submanifold is then $C = \Gamma^{\dag}\Un\perm\Um\Omega\subset\UN$.  At any $U = \Gamma^{\dag}V\perm W^{\dag}\Omega$ in this critical submanifold, the Hessian eigenvectors identified in \eqref{eqn:eigenvectors} corresponding to zero eigenvalue describe an orthonormal basis for the tangent space of $C$, $\rmT_{U}C$, and the eigenvectors corresponding to nonzero eigenvalues form an orthonormal basis for the normal space of $C$ at $U$, $(\rmT_{U}C)^{\perp}$.  While such a critical submanifold $C$ is not totally geodesic in general, it is the case that the geodesics along these particular basis directions for $\rmT_{U}C$ remain on $C$.  To see this, first note that each of these basis vectors $X$ has the property that either $[U^{\dag}X, U^{\dag}\Ob U] = 0$ or $[U^{\dag}X,\rho] = 0$.  This follows from the fact that, when $\tilde{X} = W^{\dag}\Omega U^{\dag}X\Omega^{\dag}W]\in\uN$, either $\tilde{X} = i|l\rangle\langle l|$ for some $l = 1,\dots,N$ so that $X$ commutes with both $U^{\dag}\Ob U$ and $\rho$, or $\tilde{X} = z|j\rangle \langle k| - \bar{z}|k\rangle\langle j|$ for some $|z|^{2} = 1/2$ and $j<k$ with either $\lambda_{j} = \lambda_{k}$ or $\sigma_{\perm(j)} = \sigma_{\perm(k)}$.  When $U^{\dag}X$ commutes with $U^{\dag}\Ob U$, then $\perm\tilde{X}\perm^{\dag}$ commutes with $\Sigma$, so that $\perm\tilde{X}\perm^{\dag}\in\um$.  Then the curve $\gamma(s) = e^{sXU^{\dag}}U = \Gamma^{\dag}V e^{s\perm\tilde{X}\perm^{\dag}}\perm W^{\dag}\Omega$ is a geodesic in $C$ in the direction $\gamma'(0) = X$, which can be thought of in terms of the path through $\Umn$, $(V(s),W(s)) = (Ve^{s\perm\tilde{X}\perm^{\dag}}, W)$.  Likewise, when $U^{\dag}X$ commutes with $\rho$, then $\tilde{X}$ commutes with $\Lambda$, so that $\tilde{X}\in\un$.  Then the curve $\gamma(s) = Ue^{sU^{\dag}X} = \Gamma^{\dag}V \perm e^{s\tilde{X}}W^{\dag}\Omega$ is a geodesic in $C$ in the direction $\gamma'(0) = X$, which can be thought of in terms of the path through $\Umn$, $(V(s),W(s)) = (V, W e^{-s\tilde{X}})$.  Note that, while the geodesic curve is uniquely defined for a given basis vector $X$, the corresponding path through $\Umn$ is not unique because of the ambiguity in associating $(V,W)\in\Umn$ to a given $U\in C$ afforded by the stabilizer subgroup $\Stab_{\Umn}(\perm)\cong \UK$.

For any two of these basis vectors for $\rmT_{U}C$, say $X$ and $Y$, the second fundamental form $S(X,Y)$ at $U$ is defined to be $\big(\nabla_{\hat{X}}\hat{Y}\big)_{U}^{\perp}$, where $\hat{X}$ and $\hat{Y}$ are smooth local extensions of $X$ and $Y$.  This can also be calculated by extending $Y$ along the geodesic in the direction of $X$ and taking the normal part of the covariant derivative of this field.  Letting $\tilde{Y} = W^{\dag}\Omega U^{\dag}Y\Omega^{\dag}W\in\uN$, when $U^{\dag}X$ commutes with $U^{\dag}\Ob U$, the vector $Y$ can be extended along the geodesic $\gamma(s) = e^{sXU^{\dag}}U$ by $\hat{Y}(s) = \gamma(s)\Omega^{\dag}W\tilde{Y}W^{\dag}\Omega = \Gamma^{\dag}Ve^{s\perm\tilde{X}\perm^{\dag}}\perm \tilde{Y}W^{\dag}\Omega = \Gamma^{\dag}V\perm e^{s\tilde{X}} \tilde{Y}W^{\dag}\Omega$ where the path through $\Umn$ $(V(s),W(s)) = (Ve^{s\perm\tilde{X}\perm^{\dag}}, W)$ is used to define the extension.  Likewise, when $U^{\dag}X$ commutes with $\rho$, the vector $Y$ can be extended along the geodesic $\gamma(s) = Ue^{sU^{\dag}X}$ by $\hat{Y}(s) = \Gamma^{\dag}V\perm \tilde{Y}e^{s\tilde{X}}W^{\dag}\Omega$ by using $(V(s),W(s)) = (V,W e^{-s\tilde{X}})$.  Since the paths through $\Umn$ are not uniquely defined, these extensions of $Y$ along $\gamma$ are not unique.  Indeed, this is clear in the cases where $U^{\dag}X$ commutes with both $U^{\dag}\Ob U$ and $\rho$, as the two extensions of $Y$ described here are different for any $Y$ such that $[\tilde{Y},\tilde{X}]\neq 0$.  However, since the second fundamental form is tensorial \cite{Sakai1996}, the choice of extension of $Y$ will ultimately not be important.

The covariant derivative of $\hat{Y}(s)$ can be computed simply as the ordinary derivative of $\hat{Y}$ (as a path through $\mathbb{C}^{N\times N}$) followed by projection down to the tangent space $\rmT_{\gamma(s)}\UN$.  Then, using $^{\top}$ to denote the tangential component, the covariant derivatives for the particular extensions $\hat{Y}(s)$ defined above are
	\begin{subequations}
	\begin{align}
		\nabla_{\hat{X}}\hat{Y} & = \left(\left.\frac{\rmd}{\rmd s}\hat{Y}(s)\right|_{s = 0}\right)^{\top} = \begin{cases}\left(\Gamma^{\dag}V\perm\tilde{X}\tilde{Y}W^{\dag}\Omega\right)^{\top} & \text{for } [U^{\dag}X,U^{\dag}\Ob U] = 0\\\left(\Gamma^{\dag}V\perm\tilde{Y}\tilde{X}W^{\dag}\Omega\right)^{\top} & \text{for } [U^{\dag}X, \rho] = 0\end{cases}\\
		& = \begin{cases}\frac{1}{2}\left(\Gamma^{\dag}V\perm\tilde{X}\tilde{Y}W^{\dag}\Omega - U\Omega^{\dag}W^{\dag}\tilde{Y}\tilde{X}\perm^{\dag}V^{\dag}\Gamma U\right) & \text{for } [U^{\dag}X,U^{\dag}\Ob U] = 0\\\frac{1}{2}\left(\Gamma^{\dag}V\perm\tilde{Y}\tilde{X}W^{\dag}\Omega - U\Omega^{\dag}W^{\dag}\tilde{X}\tilde{Y}\perm^{\dag}V^{\dag}\Gamma U\right) & \text{for } [U^{\dag}X, \rho] = 0\end{cases}\\
		& = \begin{cases}\frac{1}{2}\Gamma^{\dag}V\perm[\tilde{X},\tilde{Y}]W^{\dag}\Omega & \text{for } [U^{\dag}X,U^{\dag}\Ob U] = 0\\ -\frac{1}{2}\Gamma^{\dag}V\perm[\tilde{X},\tilde{Y}]W^{\dag}\Omega & \text{for }[U^{\dag}X, \rho] = 0.\end{cases}\label{eqn:covariantDerivative}
	\end{align}
	\end{subequations}
Then $S(X,Y) = \big(\nabla_{\hat{X}}\hat{Y}\big)^{\perp}$ is the normal component, where the normal space $(\rmT_{U}C)^{\perp}$ is spanned by the vectors $Z = \Gamma^{\dag}V\perm\tilde{Z}W^{\dag}\Omega$ where $\tilde{Z} = z|j\rangle\langle k| - \bar{z}|k\rangle\langle j|$ such that $|z|^{2}=1/2$ and $j$ and $k$ are such that $\lambda_{j}\neq\lambda_{k}$ and $\sigma_{\perm(j)}\neq\sigma_{\perm(k)}$.  Now, observe that if $\tilde{X} = i|l\rangle\langle l|$ and $\tilde{Y} = i|m\rangle\langle m|$, then $[\tilde{X},\tilde{Y}] = 0$.  And if $\tilde{X} = i|l\rangle\langle l|$ and $\tilde{Y} = w|j\rangle\langle k| - \bar{w}|k\rangle\langle j|$ for $j<k$, then $[\tilde{X},\tilde{Y}] = (\delta_{jl}-\delta_{kl})\big((iw)|j\rangle\langle k| - \overline{(iw)}|k\rangle\langle j|\big)$ so that $\nabla_{\hat{X}}\hat{Y}\in \rmT_{U}C$.  So whenever $X$ is such that $\tilde{X} = i|l\rangle\langle l|$, $S(X,Y) = 0$ for all $Y\in \rmT_{U}C$.  Moreover, if $\tilde{X} = v|l\rangle\langle m| - \bar{v}|m\rangle\langle l|$ and $\tilde{Y} = w|j\rangle\langle k| - \bar{w}|k\rangle\langle j|$, then 
	\begin{align}
		[\tilde{X},\tilde{Y}] & = \delta_{jm}\big(vw|l\rangle\langle k| - \bar{v}\bar{w}|k\rangle\langle l|\big) + \delta_{km}\big(\bar{v}w|j\rangle\langle l| - v\bar{w}|l\rangle\langle j|\big)\nonumber\\
		& \qquad + \delta_{kl}\big(\bar{v}\bar{w}|m\rangle\langle j| - vw|j\rangle\langle m|\big) + \delta_{jl}\big(v\bar{w}|k\rangle\langle m| - \bar{v}w|m\rangle\langle k|\big).
	\end{align}

Consequently, if $U^{\dag}X$ and $U^{\dag}Y$ both commute with $U^{\dag}\Ob U$, then $\nabla_{\hat{X}}\hat{Y}$ from \eqref{eqn:covariantDerivative} also commutes with $U^{\dag}\Ob U$ and therefore lies in $\rmT_{U}C$, so that $S(X,Y) = \big(\nabla_{\hat{X}}\hat{Y}\big)^{\perp} = 0$.  Likewise, if $U^{\dag}X$ and $U^{\dag}Y$ both commute with $\rho$, then $\nabla_{\hat{X}}\hat{Y}$ commutes with $\rho$ so that again $S(X,Y) = \big(\nabla_{\hat{X}}\hat{Y}\big)^{\perp} = 0$.  So among these basis vectors for $\rmT_{U}C$, the only pairs $X$ and $Y$ for which $S(X,Y)$ can be non-zero are those in which one of the pair $\{U^{\dag}X,U^{\dag}Y\}$ commutes with $U^{\dag}\Ob U$, but not with $\rho$, and the other commutes with $\rho$, but not with $U^{\dag}\Ob U$.  Then the basis vectors for $\rmT_{U}C$ can be divided into three categories: (1) those for which $U^{\dag}X$ commutes with $U^{\dag}\Ob U$ but not $\rho$, of which there are $
\sum m_{i}^{2} - \sum k_{ij}^{2}$; (2) those for which $U^{\dag}X$ commutes with $\rho$ but not $U^{\dag}\Ob U$, of which there are $\sum n_{j}^{2} - \sum k_{ij}^{2}$; and (3) the remaining $\sum k_{ij}^{2}$ basis vectors for which $U^{\dag}X$ commutes with both $U^{\dag}\Ob U$ and $\rho$.  

Taking any normal vector $Z\in (\rmT_{U}C)^{\perp}$, the \emph{shape operator} (or Weingarten map) \cite{Sakai1996} $A_{Z}:\rmT_{U}C\to\rmT_{U}C$, defined by $\langle A_{Z}X,Y\rangle = \langle S(X,Y),Z\rangle$, can be represented in block form using these three categories of tangent vectors as
\begin{equation}A_{Z} = \begin{bmatrix}0 & B & 0\\B^{T} & 0 & 0 \\ 0 & 0 & 0\end{bmatrix}\end{equation}
where $B$ is an $(\sum m_{i}^{2}-\sum k_{ij}^{2})\times(\sum n_{j}^{2} - \sum k_{ij}^{2})$ matrix.  The eigenvalues of $A_{Z}$ are the \emph{principal curvatures} of the submanifold $C$ at the point $U$ with respect to the normal vector $Z$.  Because of the indicated block structure, the nonzero eigenvalues of $A_{Z}$ come in positive-negative pairs $\pm\eta$, where $\eta$ is a nonzero eigenvalue of $B^{T}B$.   Moreover, these eigenvalues can be bounded by observing that the squared Hilbert-Schmidt norm $\|B\|^{2} = \sum_{j<m}\frac{|\tilde{Z}_{jm}|^{2}}{2}\big(k_{q_{m}t_{j}} + k_{p_{j}u_{m}})$ where $\lambda_{m} = \tilde{\lambda}_{q_{m}}$, $\lambda_{j} = \tilde{\lambda}_{p_{j}}$, $\sigma_{\perm(j)} = \tilde{\sigma}_{t_{j}}$, $\sigma_{\perm(m)} = \tilde{\sigma}_{u_{m}}$, and $\tilde{\lambda}_{1}>\dots>\tilde{\lambda_{s}}$ and $\tilde{\sigma}_{1}>\dots>\tilde{\sigma}_{r}$ are the distinct eigenvalues of $\rho$ and $\Ob$ with multiplicities $m_{1},\dots,m_{s}$ and $n_{1},\dots,n_{r}$.  Most importantly, $\Tr(A_{Z}) = 0$ for every normal vector $Z$, so that the mean curvature vector field \cite{Gray2004} $H = \sum S(X_{i},X_{i})$ for any orthonormal basis $\{X_{i}\}$ of $\rmT_{U}C$, is zero.  Therefore every critical submanifold of $J(U) = \Tr(U\rho U^{\dag}\Ob)$ is a \emph{minimal} submanifold of $\UN$ \cite{Sakai1996}.

\bigskip

\subsection{Examples}
\label{sec:shapeOpExamples}
We now return to the examples considered in Sections \ref{sec:volExamples} and \ref{sec:volFracExamples} and compute the shape operators of the critical submanifolds.
\begin{example}[Maximum Submanifold of $P_{i\to f}$]
	From the contingency table for $P_{i\to f} = 1$ (see Section \ref{sec:volExamples}), it is clear that $\sum m_{i}^{2} - \sum k_{ij}^{2} = \sum n_{j}^{2} - \sum k_{ij}^{2} = 0$, so that the matrix $B$ in the shape operator $A_{Z}$ is $0\times 0$.  Therefore $A_{Z}$ is always the zero operator for all normal vectors $Z$, meaning that this critical submanifold is totally geodesic \cite{Sakai1996}.
\end{example}

\begin{example}[Minimum Submanifold of $P_{i\to f}$]
	Consider the case where $\sigma_{\perm(2)} = 1$ and $\sigma_{\perm(j)} = 0$ for all $j\neq 2$.  The first category basis vectors are then $X_{j} = \Gamma^{\dag}V\perm\tilde{X}_{j}W^{\dag}\Omega$, where $\tilde{X}_{j} = w_{j}|1\rangle \langle j| - \bar{w}_{j}|j\rangle\langle 1|$ for $j=3,\dots,N$ and for $w_{j} = \frac{1}{\sqrt{2}}$ and $w_{j} = \frac{i}{\sqrt{2}}$.  Likewise, the second category basis vectors are $Y_{j} = \Gamma^{\dag}V\perm\tilde{Y}_{j}W^{\dag}\Omega$ for $\tilde{Y}_{j} = v_{j}|2\rangle \langle j| - \bar{v}_{j}|j\rangle\langle 2|$ for $j=3,\dots,N$, and for $v_{j} = \frac{1}{\sqrt{2}}$ and $v_{j} = \frac{i}{\sqrt{2}}$.  So $\nabla_{\hat{X}_{j}}\hat{Y}_{k} = \delta_{jk}\frac{1}{2}\Gamma^{\dag}V\perm[\tilde{X}_{j},\tilde{Y}_{j}]W^{\dag}\Omega$, where $[\tilde{X}_{j},\tilde{Y}_{j}] = -w_{j}\bar{v}_{j}|1\rangle\langle 2| + \bar{w}_{j}v_{j}|2\rangle\langle 1|$. Then for any normal vector $Z = \Gamma^{\dag}V\perm\tilde{Z}W^{\dag}\Omega$ with $\tilde{Z} = z|1\rangle\langle 2| - \bar{z}|2\rangle\langle 1|$, $\langle \nabla_{\hat{X}}\hat{Y}, Z\rangle =\frac{1}{2}\langle [\tilde{X}, \tilde{Y}], \tilde{Z}\rangle = -\Re(\bar{w}_{j}v_{j}z)$, so the $(2N-4)\times(2N-4)$ matrix $B$ in the shape operator $A_{Z}$ is 
	\begin{equation}B = -\frac{1}{2}\mathbb{I}_{N-2}\otimes\begin{bmatrix}\Re(z) & \Im(z)\\-\Im(z) & \Re(z)\end{bmatrix}.\end{equation}  Then $B^{T}B = \frac{|z|^{2}}{4}\mathbb{I}_{2N-4}$, so that $A_{Z}$ has eigenvalue $\eta = \frac{|z|}{2}$ with multiplicity $2N-4$, $\eta = -\frac{|z|}{2}$ with multiplicity $2N-4$, and $\eta = 0$ with multiplicity $N^{2}-4N+6$.  When $Z$ is normalized with respect to the Hilbert-Schmidt norm, $|z|^{2} = 1/2$, so the nonzero eigenvalues of $A_{Z}$ are $\eta = \pm\frac{1}{2\sqrt{2}}$.
\end{example}

\begin{example}[Fully Non-degenerate $\rho$ and $\Ob$]
	In the case where $\rho$ and $\Ob$ are both fully non-degenerate, $m_{1} = \dots = m_{N} = 1$, $n_{1} = \dots = n_{N} = 1$ and $\mathbf{K} = \perm^{\dag}$.  Then for any critical submanifold $\Orb(\perm)$, $\sum m_{i}^{2} - \sum k_{ij}^{2} = \sum n_{j}^{2} - \sum k_{ij}^{2} = 0$, so that the matrix $B$ in the shape operator $A_{Z}$ is $0\times 0$.  Therefore $A_{Z}$ is always the zero operator for all normal vectors $Z$, so that every critical submanifold of $J$ is totally geodesic when $\rho$ and $\Ob$ are fully non-degenerate.
\end{example}

\bigskip

\section{In Support of Conjecture \ref{conj:gradientBound}}
\label{sec:supportConj}
In this appendix we build the case for Conjecture \ref{conj:gradientBound}.  Let $U = \Gamma^{\dag}V\perm W^{\dag}\Omega$ be a critical point, where $(V,W)\in\Umn$ and $\Omega,\Gamma\in\UN$ diagonalize $\rho$ and $\Ob$, respectively, i.e. $\Omega\rho\Omega^{\dag} = \Lambda$ and $\Gamma\Ob \Gamma^{\dag} = \Sigma$.  Also let $\tilde{A} = \sum_{l=1}^{L}\alpha_{l}\tilde{A}_{l}$, where $\sum \alpha_{l}^{2} = 1$, $\tilde{A}_{l} = z_{l}|j_{l}\rangle\langle k_{l}| - \bar{z}_{l}|k_{l}\rangle\langle j_{l}|$, $|z_{l}|^{2} = 1/2$, and $\{j_{1},k_{2},\dots, j_{L},k_{L}\}$ is a set of $2L$ distinct indices in $\{1,\dots, N\}$ such that the Hessian eigenvalue $\beta_{j_{l}k_{l}} = -(\lambda_{j_{l}}- \lambda_{k_{l}})(\sigma_{\perm(j_{l})} - \sigma_{\perm(k_{l})})$ is nonzero.  As shown in Appendix \ref{sec:HessianStructure}, for the Hessian at a critical point, every eigenspace corresponding to a nonzero eigenvalue is spanned by eigenvectors of the form $z|j\rangle \langle k|-\bar{z}|k\rangle\langle j|$ with $j < k$.  Then $\exp(sA)\rho\exp(-sA) = \Omega^{\dag}W\exp(s\tilde{A})\Lambda\exp(-s\tilde{A})W^{\dag}\Omega$.  Now, defining $\tilde{A}_{l}^{+} = z_{l}|j_{l}\rangle\langle k_{l}| + \bar{z}_{l}|k_{l}\rangle\langle j_{l}|$ and $\tilde{A}^{+} = \sum \alpha_{l}\tilde{A}_{l}^{+}$, it is easily verified that 
	\begin{subequations}
	\begin{align}
		[\tilde{A},\Lambda] & = \sum_{l=1}^{L}(\lambda_{k_{l}}-\lambda_{j_{l}})\alpha_{l}\tilde{A}_{l}^{+}\\
		[\tilde{A}, \tilde{A}_{l}^{+}] & = \alpha_{l}\big(|j_{l}\rangle\langle j_{l}| - |k_{l}\rangle\langle k_{l}|\big)\\
		[\tilde{A}, |j_{l}\rangle\langle j_{l}| - |k_{l}\rangle\langle k_{l}|] & = -2\alpha_{l}\tilde{A}_{l}^{+},
	\end{align}
	\end{subequations}
so that $[\tilde{A},[\tilde{A}, \tilde{A}_{l}^{+}]] = -2\alpha_{l}^{2}\tilde{A}_{l}^{+}$, and $[\tilde{A},[\tilde{A}, |j_{l}\rangle\langle j_{l}| - |k_{l}\rangle\langle k_{l}|]] = -2\alpha_{l}^{2}\big(|j_{l}\rangle\langle j_{l}| - |k_{l}\rangle\langle k_{l}|\big)$.  Then, using the notation $\ad_{X}$ for the adjoint operator $\ad_{X}(Y) = [X,Y]$ \cite{Hall2003}, 
	\begin{subequations}
	\begin{align}
		\exp(s\tilde{A})\Lambda\exp(-s\tilde{A}) & = \exp(s\ad_{\tilde{A}})\Lambda = \sum_{m=0}^{\infty}\frac{s^{m}\ad_{\tilde{A}}^{m}(\Lambda)}{m!}\\
		& = \Lambda + \sum_{m=1}^{\infty}\frac{s^{2m}\ad_{\tilde{A}}^{2m}(\Lambda)}{(2m)!} + \sum_{m=0}^{\infty}\frac{s^{2m+1}\ad_{\tilde{A}}^{2m+1}(\Lambda)}{(2m+1)!}\\
		& = \Lambda + \sum_{l=1}^{L}\alpha_{l}^{2}(\lambda_{k_{l}}-\lambda_{j_{l}})\sum_{m=1}^{\infty}\frac{s^{2m}\ad_{\tilde{A}}^{2m-2}\big(|j_{l}\rangle\langle j_{l}|-|k_{l}\rangle\langle k_{l}|\big)}{(2m)!}\nonumber\\
		& \qquad + \sum_{l=1}^{L}\alpha_{l}(\lambda_{k_{l}}-\lambda_{j_{l}})\sum_{m=0}^{\infty}\frac{s^{2m+1}\ad_{\tilde{A}}^{2m}(\tilde{A}_{l}^{+})}{(2m+1)!}\\
		& = \Lambda + \sum_{l=1}^{L}(\lambda_{k_{l}}-\lambda_{j_{l}})\sum_{m=1}^{\infty}\frac{(\alpha_{l}s)^{2m}(-2)^{m-1}}{(2m)!}\big(|j_{l}\rangle\langle j_{l}|-|k_{l}\rangle\langle k_{l}|\big)\nonumber\\
		& \qquad + \sum_{l=1}^{L}(\lambda_{k_{l}}-\lambda_{j_{l}})\sum_{m=0}^{\infty}\frac{(\alpha_{l}s)^{2m+1}(-2)^{m}}{(2m+1)!}\tilde{A}_{l}^{+}\\
		& = \Lambda + \sum_{l}\frac{\lambda_{k_{l}}-\lambda_{j_{l}}}{2}\big(1-\cos(\sqrt{2}\alpha_{l}s)\big)\big(|j_{l}\rangle\langle j_{l}|-|k_{l}\rangle\langle k_{l}|\big)\nonumber\\
		& \qquad + \sum_{l}\frac{\lambda_{k_{l}}-\lambda_{j_{l}}}{\sqrt{2}}\sin(\sqrt{2}\alpha_{l}s)\tilde{A}_{l}^{+}
	\end{align}
	\end{subequations}
so that
\begin{subequations}
\begin{align}
	[U^{\dag}\Ob U, e^{sA}\rho e^{-sA}] & = \Omega^{\dag}W[\perm^{\dag}\Sigma\perm, \exp(s\tilde{A})\Lambda\exp(-s\tilde{A})]W^{\dag}\Omega\\
	& = \Omega^{\dag}W\Big([\perm^{\dag}\Sigma\perm, \Lambda] + \sum_{l}\frac{\lambda_{k_{l}}-\lambda_{j_{l}}}{2}\big(1-\cos(\sqrt{2}\alpha_{l}s)\big)\big[\perm^{\dag}\Sigma\perm,|j_{l}\rangle\langle j_{l}|-|k_{l}\rangle\langle k_{l}|\big]\nonumber\\
	& \qquad\qquad \mbox{} + \sum_{l}\frac{\lambda_{k_{l}}-\lambda_{j_{l}}}{\sqrt{2}}\sin(\sqrt{2}\alpha_{l}s)\big[\perm^{\dag}\Sigma\perm, \tilde{A}_{l}^{+}\big]\Big)W^{\dag}\Omega\\
	& = \sum_{l}\frac{(\lambda_{k_{l}}-\lambda_{j_{l}})(\sigma_{\perm(k_{l})} - \sigma_{\perm(j_{l})})}{\sqrt{2}}\sin(\sqrt{2}\alpha_{l}s)A_{l}\\
	& = -\sum_{l}\frac{\beta_{j_{l}k_{l}}}{\sqrt{2}}\sin(\sqrt{2}\alpha_{l}s)A_{l}
\end{align}
\end{subequations}
where $\beta_{jk}$ is the Hessian eigenvalue corresponding to the eigenvector $U\big(z|j\rangle \langle k| - \bar{z}|k\rangle\langle j|\big)$.  It follows that, for this choice of $A$, $f(s) = \sum_{l}\beta_{j_{l}k_{l}}^{2}\sin^{2}(\sqrt{2}\alpha_{l}s)/2$ for all $s$.  Now, let $q(s) = \sin^{2}(\sqrt{2}\alpha_{l}s) - \alpha_{l}^{2}\sin^{2}(\sqrt{2}s)$.  Then $q'(s) = \sqrt{2}\alpha_{l}\sin(2\sqrt{2}\alpha_{l}s) - \sqrt{2}\alpha_{l}^{2}\sin(2\sqrt{2}s)$, and $q''(s) = 4\alpha_{l}^{2}\big[\cos(2\sqrt{2}\alpha_{l}s) - \cos(2\sqrt{2}s)\big]$.  Since $q'(0) = 0$ and $q''(s)\geq 0$ for all $s\in[0,\pi/(2\sqrt{2})]$, $q'(s)\geq 0$ for all $s\in[0,\pi/(2\sqrt{2})]$.  And since $q(0) = 0$, this implies that $q(s)\geq 0$ for all $s\in[0,\pi/(2\sqrt{2})]$.  Therefore, $f(s)\geq \sum_{l}\beta_{j_{l}k_{l}}^{2}\alpha_{l}^{2}\sin^{2}(\sqrt{2}s)/2 \geq \beta_{\min}^{2}\sin^{2}(\sqrt{2}s)/2$, since $\sum\alpha_{l}^{2} = 1$.  So for all unit normal vectors $UA$ of this form, $f(s)$ satisfies the conjecture.  Moreover, when $\tilde{A} = z|j\rangle\langle k| - \bar{z}|k\rangle\langle j|$ for some $j\neq k$, $f(s) = \beta_{jk}^{2}\sin^{2}(\sqrt{2}s)/2$, and in the particular case when $j$ and $k$ are such that $UA$ is an eigenvector corresponding to the smallest eigenvalue (in absolute value), then $f(s) = \beta_{\min}^{2}\sin^{2}(\sqrt{2}s)/2$ for all $s$, so this conjectured lower bound is attained.

Now, for that same critical point $U = \Gamma^{\dag}V\perm W^{\dag}\Omega$, let $ = \mathcal{N}(\Hess)^{\perp}\subset \rmT_{U}\UN$ be the orthogonal complement of the null space of the Hessian at $U$, i.e. the span of the Hessian eigenvectors corresponding to non-zero eigenvalues.  Then let $\mathscr{S}\subset U^{\dag}Z\subset\uN$ be the unit sphere within $U^{\dag}Z$.  For any $s\in[0,\pi/(2\sqrt{2})]$, let $\zeta_{s}:\mathscr{S}\to\mathbb{R}$ be given by $\zeta_{s}(A) = \big\|\grad J\big(U\exp(sA)\big)\big\|^{2}$.  Then, using the integral expression for the derivative of the matrix exponential \cite[App. B]{Snider1964}, the differential of $\zeta_{s}$ is found to be
\begin{subequations}
\begin{align}
	\rmd_{A}\zeta_{s}(\delta A) & = 2\left\langle\left[e^{-sA}U^{\dag}\Ob U e^{sA}, \rho\right], \; \left[\left[e^{-sA}U^{\dag}\Ob Ue^{sA}, s\int_{0}^{1}e^{-srA}\delta A e^{srA}\,\rmd r\right],\rho\right]\right\rangle\\
	& = 2s\left\langle\int_{0}^{1}e^{srA}\left[e^{-sA}U^{\dag}\Ob Ue^{sA}, \left[\left[e^{-sA}U^{\dag}\Ob U e^{sA}, \rho\right],\rho\right]\right]e^{-srA}\,\rmd r, \;\delta A \right\rangle\\
	& = 2s\left\langle\int_{0}^{1}e^{-s(1-r)\tilde{A}}\left[\perm^{\dag}\Sigma\perm , \left[\left[ \perm^{\dag}\Sigma\perm , e^{s\tilde{A}}\Lambda e^{-s\tilde{A}}\right],e^{s\tilde{A}}\Lambda e^{-s\tilde{A}}\right]\right] e^{s(1-r)\tilde{A}}\,\rmd r, \;W^{\dag}\Omega\delta A\Omega^{\dag}W \right\rangle,\label{eqn:gradzeta}
\end{align} 
\end{subequations}
where $\tilde{A} = W^{\dag}\Omega A\Omega^{\dag}W$.  When $\tilde{A} = \sum_{l}\alpha_{l}\tilde{A}_{l}$ with $\tilde{A}_{l} = z_{l}|j_{l}\rangle\langle k_{l} - \bar{z}_{l}|k_{l}\rangle\langle j_{l}|$ as above, we have seen that $\big[ \perm^{\dag}\Sigma\perm , e^{s\tilde{A}}\Lambda e^{-s\tilde{A}}\big] = \sum_{l}\beta_{j_{l}k_{l}}\sin(\sqrt{2}\alpha_{l}s)\tilde{A}_{l}/\sqrt{2}$.  Then 
\begin{subequations}
\begin{align}
	\Big[\Big[ \perm^{\dag}\Sigma\perm , & e^{s\tilde{A}}\Lambda e^{-s\tilde{A}}\Big],  e^{s\tilde{A}}\Lambda e^{-s\tilde{A}}\Big]\nonumber\\
	& = \sum_{l=1}^{L}\frac{\beta_{j_{l}k_{l}}}{\sqrt{2}}\sin(\sqrt{2}\alpha_{l}s)e^{s\tilde{A}}\big[\tilde{A},\Lambda\big]e^{-s\tilde{A}}\\
	& = \sum_{l=1}^{L}\alpha_{l}(\lambda_{k_{l}} - \lambda_{j_{l}})\frac{\beta_{j_{l}k_{l}}}{\sqrt{2}}\sin(\sqrt{2}\alpha_{l}s)e^{s\tilde{A}}\tilde{A}_{l}^{+}e^{-s\tilde{A}}\\
	& = \sum_{l=1}^{L}\alpha_{l}(\lambda_{k_{l}} - \lambda_{j_{l}})\frac{\beta_{j_{l}k_{l}}}{\sqrt{2}}\sin(\sqrt{2}\alpha_{l}s)\sum_{m=0}^{\infty}\frac{s^{m}\ad_{\tilde{A}}^{m}(\tilde{A}_{l}^{+})}{m!}\\
	& = \sum_{l=1}^{L}\alpha_{l}(\lambda_{k_{l}} - \lambda_{j_{l}})\frac{\beta_{j_{l}k_{l}}}{\sqrt{2}}\sin(\sqrt{2}\alpha_{l}s)\sum_{m=0}^{\infty}\frac{s^{2m}\ad_{\tilde{A}}^{2m}(\tilde{A}_{l}^{+})}{(2m)!}\nonumber\\
	& \qquad \mbox{} + \sum_{l=1}^{L}\alpha_{l}(\lambda_{k_{l}} - \lambda_{j_{l}})\frac{\beta_{j_{l}k_{l}}}{\sqrt{2}}\sin(\sqrt{2}\alpha_{l}s)\sum_{m=0}^{\infty}\frac{s^{2m+1}\ad_{\tilde{A}}^{2m+1}(\tilde{A}_{l}^{+})}{(2m+1)!}\\
	& = \sum_{l=1}^{L}\alpha_{l}(\lambda_{k_{l}} - \lambda_{j_{l}})\frac{\beta_{j_{l}k_{l}}}{\sqrt{2}}\sin(\sqrt{2}\alpha_{l}s)\sum_{m=0}^{\infty}\frac{(\alpha_{l}s)^{2m}(-2)^{m}}{(2m)!}\tilde{A}_{l}^{+}\nonumber\\
	& \qquad \mbox{} + \sum_{l=1}^{L}\alpha_{l}(\lambda_{k_{l}} - \lambda_{j_{l}})\frac{\beta_{j_{l}k_{l}}}{\sqrt{2}}\sin(\sqrt{2}\alpha_{l}s)\sum_{m=0}^{\infty}\frac{(\alpha_{l}s)^{2m+1}(-2)^{m}}{(2m+1)!}\big(|j_{l}\rangle\langle j_{l}| - |k_{l}\rangle\langle k_{l}|\big)\\
	& = \sum_{l=1}^{L}\alpha_{l}(\lambda_{k_{l}} - \lambda_{j_{l}})\frac{\beta_{j_{l}k_{l}}}{2}\left\{\frac{1}{\sqrt{2}}\sin(2\sqrt{2}\alpha_{l}s)\tilde{A}_{l}^{+} + \sin^{2}(\sqrt{2}\alpha_{l}s)\big(|j_{l}\rangle\langle j_{l}| - |k_{l}\rangle\langle k_{l}|\big)\right\},
%	& = \sum_{l=1}^{L}\alpha_{l}(\lambda_{k_{l}} - \lambda_{j_{l}})\frac{\beta_{j_{l}k_{l}}}{2\sqrt{2}}\left\{\frac{1}{\sqrt{2}}\sin(2\sqrt{2}\alpha_{l}s)\tilde{A}_{l}^{+} + (\lambda_{k_{l}} - \lambda_{j_{l}})\frac{\beta_{j_{l}k_{l}}}{2}\sin^{2}(\sqrt{2}\alpha_{l}s)\big(|j_{l}\rangle\langle j_{l}| - |k_{l}\rangle\langle k_{l}|\big),
\end{align}
\end{subequations}
so that 
\begin{subequations}
\begin{align}
	\int_{0}^{1}e^{-s(1-r)\tilde{A}} & \left[\perm^{\dag}\Sigma\perm , \left[\left[ \perm^{\dag}\Sigma\perm , e^{s\tilde{A}}\Lambda e^{-s\tilde{A}}\right],e^{s\tilde{A}}\Lambda e^{-s\tilde{A}}\right]\right] e^{s(1-r)\tilde{A}}\,\rmd r \nonumber\\
	& = -\sum_{l=1}^{L}\alpha_{l}\frac{\beta_{j_{l}k_{l}}^{2}}{2\sqrt{2}}\sin(2\sqrt{2}\alpha_{l}s)\int_{0}^{1}e^{-s(1-r)\tilde{A}}\tilde{A}_{l} e^{s(1-r)\tilde{A}}\,\rmd r\\
	& = -\sum_{l=1}^{L}\alpha_{l}\frac{\beta_{j_{l}k_{l}}^{2}}{2\sqrt{2}}\sin(2\sqrt{2}\alpha_{l}s)\tilde{A}_{l},\label{eqn:constrainedGradzeta}
\end{align}
\end{subequations}
because $\tilde{A}_{l}$ commutes with $\tilde{A}$ for all $l$ under the stated assumptions.  Since the domain of $\zeta_{s}$ is the unit sphere $\mathscr{S}\subset U^{\dag}Z\subset \un$, $A$ is is critical point of $\zeta_{s}$ if and only if $\rmd_{A}\zeta_{s}(\delta A)$ is zero for all $\delta A$ perpendicular to $A$, i.e. if and only if \eqref{eqn:constrainedGradzeta} is proportional to $\tilde{A} = \sum \alpha_{l}\tilde{A}_{l}$. So for a given $s$ there exist many critical points of $\zeta_{s}$ among these $A$ matrices, in particular the cases in which $\tilde{A} = z|j\rangle \langle k| - \bar{z}|k|\rangle\langle j|$ are critical for all $s$.  If a more detailed critical point analysis could be performed for $\zeta_{s}$ with $s\in[0,\pi/(2\sqrt{2})]$, it might lead to a confirmation of Conjecture \ref{conj:gradientBound}.

In addition to these analytic results that hint at the conjecture, more than one million numerical simulations have been performed to test it.  In each simulation, random degeneracy structures were sampled for $\rho$ and $\Ob$, random eigenvalues chosen, a random critical point and a random normal vector selected.  Then $f(s)$ was computed over the interval $[0,\pi/(2\sqrt{2})]$ and compared to the function $\beta_{\min}^{2}\sin^{2}(\sqrt{2}s)/2$.  This was done for different system sizes including $N=6$, $N=17$, $N=25$, and 1000 trials with $N=256$ (which take much longer to run).  In every case Conjecture \ref{conj:gradientBound} was satisfied.  So seemingly, either the conjecture is true, or any counterexamples that exist are hard to find.

\bigskip

\bibliographystyle{unsrtnat}
%\bibliography{/users/dominyjm/Documents/bibtex_qcdb}
\bibliography{../TrpO_nearcritical}

\begin{thebibliography}{21}
\providecommand{\natexlab}[1]{#1}
\providecommand{\url}[1]{\texttt{#1}}
\expandafter\ifx\csname urlstyle\endcsname\relax
  \providecommand{\doi}[1]{doi: #1}\else
  \providecommand{\doi}{doi: \begingroup \urlstyle{rm}\Url}\fi

\bibitem[Chakrabarti and Rabitz(2007)]{Chakrabarti2007}
Raj Chakrabarti and Herschel Rabitz.
\newblock Quantum control landscapes.
\newblock \emph{Int. Rev. Phys. Chem.}, 26\penalty0 (4):\penalty0 671 -- 735,
  Oct. 2007.
\newblock \doi{10.1080/01442350701633300}.
\newblock URL \url{http://www.informaworld.com/10.1080/01442350701633300}.

\bibitem[Rabitz et~al.(2004)Rabitz, Hsieh, and Rosenthal]{Rabitz2004}
Herschel~A. Rabitz, Michael~M. Hsieh, and Carey~M. Rosenthal.
\newblock Quantum optimally controlled transition landscapes.
\newblock \emph{Science}, 303:\penalty0 1998--2001, Mar. 26 2004.
\newblock \doi{10.1126/science.1093649}.
\newblock URL \url{http://dx.doi.org/10.1126/science.1093649}.

\bibitem[Rabitz et~al.(2006{\natexlab{a}})Rabitz, Ho, Hsieh, Kosut, and
  Demiralp]{Rabitz2006}
Herschel~A. Rabitz, Tak-San Ho, Michael~M. Hsieh, Robert Kosut, and Metin
  Demiralp.
\newblock Topology of optimally controlled quantum mechanical transition
  probability landscapes.
\newblock \emph{Phys. Rev. A}, 74:\penalty0 012721, 2006{\natexlab{a}}.
\newblock \doi{10.1103/PhysRevA.74.012721}.
\newblock URL \url{http://link.aps.org/abstract/PRA/v74/e012721}.

\bibitem[Hsieh et~al.(2008)Hsieh, Wu, Rosenthal, and Rabitz]{Hsieh2008}
Michael Hsieh, Rebing Wu, Carey Rosenthal, and Herschel Rabitz.
\newblock Topological and statistical properties of quantum control transition
  landscapes.
\newblock \emph{J. Phys. B: At. Mol. Opt. Phys.}, 41\penalty0 (7):\penalty0
  074020, 2008.
\newblock \doi{10.1088/0953-4075/41/7/074020}.
\newblock URL \url{http://stacks.iop.org/0953-4075/41/i=7/a=074020}.

\bibitem[Rabitz et~al.(2006{\natexlab{b}})Rabitz, Hsieh, and
  Rosenthal]{Rabitz2006a}
Herschel~A. Rabitz, Michael~M. Hsieh, and Carey~M. Rosenthal.
\newblock Optimal control landscapes for quantum observables.
\newblock \emph{J. Chem. Phys.}, 124:\penalty0 204107, 2006{\natexlab{b}}.
\newblock \doi{10.1063/1.2198837}.
\newblock URL \url{http://link.aip.org/link/doi/10.1063/1.2198837}.

\bibitem[Ho and Rabitz(2006)]{Ho2006}
Tak-San Ho and Herschel Rabitz.
\newblock Why do effective quantum controls appear easy to find?
\newblock \emph{Journal of Photochemistry and Photobiology}, 180:\penalty0
  226--240, 2006.
\newblock \doi{10.1016/j.jphotochem.2006.03.038}.
\newblock URL \url{http://dx.doi.org/10.1016/j.jphotochem.2006.03.038}.

\bibitem[Wu et~al.(2008)Wu, Rabitz, and Hsieh]{Wu2008}
Rebing Wu, Herschel Rabitz, and Michael Hsieh.
\newblock Characterization of the critical submanifolds in quantum ensemble
  control landscapes.
\newblock \emph{J. Phys. A: Math. Theor.}, 41:\penalty0 015006, 2008.
\newblock \doi{10.1088/1751-8113/41/1/015006}.
\newblock URL \url{http://stacks.iop.org/1751-8121/41/i=1/a=015006}.

\bibitem[Rabitz et~al.(2005)Rabitz, Hsieh, and Rosenthal]{Rabitz2005}
Herschel~A. Rabitz, Michael~M. Hsieh, and Carey~M. Rosenthal.
\newblock Landscape for optimal control of quantum-mechanical unitary
  transformations.
\newblock \emph{Phys. Rev. A}, 72:\penalty0 052337, 2005.
\newblock \doi{10.1103/PhysRevA.72.052337}.
\newblock URL \url{http://link.aps.org/abstract/PRA/v72/e052337}.

\bibitem[Hsieh and Rabitz(2008)]{Hsieh2008a}
Michael Hsieh and Herschel Rabitz.
\newblock Optimal control landscape for the generation of unitary
  transformations.
\newblock \emph{Phys. Rev. A}, 77:\penalty0 042306, 2008.
\newblock \doi{10.1103/PhysRevA.77.042306}.
\newblock URL \url{http://link.aps.org/doi/10.1103/PhysRevA.77.042306}.

\bibitem[Ho et~al.(2009)Ho, Dominy, and Rabitz]{Ho2009}
Tak-San Ho, Jason Dominy, and Herschel Rabitz.
\newblock The landscape of unitary transformations in controlled quantum
  dynamics.
\newblock \emph{Phys. Rev. A}, 79:\penalty0 013422, 2009.
\newblock \doi{10.1103/PhysRevA.79.013422}.
\newblock URL \url{http://link.aps.org/doi/10.1103/PhysRevA.79.013422}.

\bibitem[Steenrod(1951)]{Steenrod1951}
Norman Steenrod.
\newblock \emph{The Topology of Fibre Bundles}.
\newblock Princeton University Press, Princeton, 1951.

\bibitem[Warner(1983)]{Warner1983}
Frank~W. Warner.
\newblock \emph{Foundations of Differentiable Manifolds and {L}ie Groups}.
\newblock Springer, New York, 1983.

\bibitem[Macdonald(1980)]{Macdonald1980}
I.~G. Macdonald.
\newblock The volume of a compact {L}ie group.
\newblock \emph{Invent. Math.}, 56\penalty0 (2):\penalty0 93--95, 1980.
\newblock \doi{10.1007/BF01392542}.
\newblock URL \url{http://dx.doi.org/10.1007/BF01392542}.

\bibitem[Hashimoto(1997)]{Hashimoto1997}
Y.~Hashimoto.
\newblock On {M}acdonald's formula for the volume of a compact {L}ie group.
\newblock \emph{Comment. Math. Helv.}, 72\penalty0 (4):\penalty0 660--662,
  1997.
\newblock \doi{10.1007/s000140050040}.
\newblock URL \url{http://dx.doi.org/10.1007/s000140050040}.

\bibitem[Weyl(1939)]{Weyl1939}
Hermann Weyl.
\newblock On the volume of tubes.
\newblock \emph{Amer. J. Math.}, 61\penalty0 (2):\penalty0 461--472, 1939.
\newblock URL \url{http://www.jstor.org/stable/2371513}.

\bibitem[Gray(2004)]{Gray2004}
Alfred Gray.
\newblock \emph{Tubes}.
\newblock Birkh{\"a}user, second edition, 2004.

\bibitem[do~Carmo(1992)]{doCarmo1992}
Manfredo~Perdig{\~a}o do~Carmo.
\newblock \emph{Riemannian Geometry}.
\newblock Birkh{\"a}user, Boston, 1992.

\bibitem[Humphreys(1972)]{Humphreys1972}
James~E. Humphreys.
\newblock \emph{Introduction to {L}ie Algebras and Representation Theory}.
\newblock Springer, New York, 1972.

\bibitem[Hall(2003)]{Hall2003}
Brian~C. Hall.
\newblock \emph{{L}ie Groups, {L}ie Algebras and Representations: An Elementary
  Introduction}.
\newblock Springer, New York, 2003.

\bibitem[Sakai(1996)]{Sakai1996}
Takashi Sakai.
\newblock \emph{Riemannian Geometry}, volume 149 of \emph{Translations of
  Mathematical Monographs}.
\newblock AMS, Providence, RI, 1996.

\bibitem[Snider(1964)]{Snider1964}
R.~F. Snider.
\newblock Perturbation variation methods for a quantum boltzmann equation.
\newblock \emph{J. Math. Phys.}, 5:\penalty0 1580, 1964.
\newblock \doi{10.1063/1.1931191}.
\newblock URL \url{http://link.aip.org/link/doi/10.1063/1.1931191}.

\end{thebibliography}

\end{document}